\documentclass[12pt]{article}
\usepackage{enumerate}

\usepackage{amsmath,amsfonts,amsthm,amssymb,graphicx,natbib,geometry,arydshln,umoline,subfig,newtxtext,newtxmath,comment}
\usepackage[title]{appendix}
\usepackage{color, comment}

\newtheorem{lemma}{Lemma}

\newtheorem{proposition}{Proposition}
\newtheorem{corollary}[proposition]{Corollary}

\theoremstyle{definition}
\newtheorem{definition}{Definition}

\theoremstyle{remark}

\def\E{{\mathrm{E}}}

\title{Optimal and Robust Disclosure of Public Information%
\thanks{This paper builds on an earlier version entitled ``Optimal Disclosure of Public Information with Endogenous Acquisition of Private Information,'' which focused on optimal disclosure.
I thank seminar participants for valuable discussions and comments at the University of Tokyo, Shanghai University of Finance and Economics, Otaru University of Commerce, Roy-Adres Seminar Economic Theory, and Econometric Society 2015 World Congress. 
I acknowledge financial support by MEXT, Grant-in-Aid for Scientific Research (15K03348, 18H05217). }} 

\author{Takashi Ui
\\Hitotsubashi University\\\texttt{oui@econ.hit-u.ac.jp}}

\date{March 2022}

\begin{document}

\maketitle

\begin{abstract}
A policymaker discloses public information to interacting agents who also acquire costly private information.  
More precise public information reduces the precision and cost of acquired private information. 
Considering this effect, what disclosure rule should the policymaker adopt? We address this question under two alternative assumptions using a linear-quadratic Gaussian game with arbitrary quadratic material welfare and convex information costs. 
First, the policymaker knows the cost of private information and adopts an optimal disclosure rule to maximize the expected welfare. 
Second, the policymaker is uncertain about the cost and adopts a robust disclosure rule to maximize the worst-case welfare. 
Depending on the elasticity of marginal cost,  
an optimal rule is qualitatively the same as that in the case of either a linear information cost or exogenous private information. 
The worst-case welfare is strictly increasing if and only if full disclosure is optimal under some information costs, which provides a new rationale for central bank transparency. 

\end{abstract}

\newpage

\section{Introduction}

Consider a policymaker (such as a central bank) who discloses public information to interacting agents (such as firms and consumers) who also acquire costly private information. 
The policymaker's concern is social welfare, including the agents' cost of information acquisition.  
When the policymaker provides more precise public information, the agents have less incentive to acquire private information, reducing its precision and cost. 
This effect of public information is referred to as the crowding-out effect \citep{colombofemminis2014}. 
Less private information can be harmful to welfare, but less information cost is beneficial; that is, the welfare implication of the crowding-out effect is unclear. 
Then, what disclosure rule should the policymaker adopt?

We address this question under two alternative assumptions. 
First, the policymaker correctly anticipates the agents' acquisition of private information. 
A disclosure rule is optimal if it maximizes the expected welfare.
Second, the policymaker is uncertain about the precision and cost of acquired private information.  
A disclosure rule is robust if it maximizes the worst-case welfare. 
This article obtains both rules and discusses the difference between them.

Our model is a three-stage game based on a symmetric linear-quadratic Gaussian (LQG) game with a continuum of agents, where a payoff function is quadratic and an information structure is Gaussian.  
In period 1, the policymaker chooses the precision of public information. 
Social welfare (i.e., the policymaker's objective function) is a material benefit minus the agents' information cost, and the material benefit is an arbitrary quadratic function such as the aggregate payoff. 
In period~2, each agent chooses the precision of (conditionally independent) private information. The cost of information is increasing and convex in the precision of private information. 
In period 3, each agent observes private and public signals and chooses an action. 
The last-period subgame is the game studied by \citet{angeletospavan2007}, and the second-period subgame is the game studied by \citet{colombofemminis2014}. 

In our analysis, volatility (the variance of the average action) and dispersion (the variance of individual actions around the average action) play an essential role in two ways. 
First, the welfare (in the equilibrium of the second-period subgame) is represented as a linear combination of the volatility and the dispersion minus the cost of information. 
Next, our key tool is marginal welfare with respect to dispersion (MWD). 
The MWD is the change in the welfare resulting from an increase in the precision of public information that induces a one-unit decrease in the dispersion. 
Indeed, the dispersion decreases with public information because it equals the difference between the variance and the covariance of individual actions, and more precise public information brings them closer. 
In contrast, the volatility increases with public information because it equals the covariance of individual actions, and more precise public information increases the correlation coefficient and thus the covariance. 

The MWD has a useful representation: it equals the weighted average of the MWDs under a linear information cost and under exogenous private information, where the relative weights are one and the elasticity of marginal information cost, respectively. 
The elasticity of marginal cost measures the degree of convexity in the cost function. 
It takes the minimum value of zero when the cost function is linear, in which case the crowding-out effect is largest. 
Therefore, the social value of public information has the same sign as that under the largest crowding-out effect 
if the elasticity is small enough and under no crowding-out effect otherwise. 

This result leads us to identify an optimal disclosure rule, which qualitatively coincides with the case of either the largest crowding-out effect (a linear information cost) or no crowding-out effect (exogenous private information), depending on the elasticity of marginal cost. 
In both cases, the following holds. 
The welfare necessarily increases with public information if the volatility's coefficient is positive and relatively large compared to the dispersion's one, in which case full disclosure is optimal.
On the other hand, the welfare necessarily decreases with public information if the volatility's coefficient is negative and relatively small compared to the dispersion's one, in which case no disclosure is optimal.
This result is attributed to the fact that more precise public information increases the volatility and decreases the dispersion.  
In other cases, partial disclosure can also be optimal. 


We next consider a robust disclosure rule. The policymaker is uncertain about the agents' cost function and evaluates the precision using the worst-case welfare, which is the infimum of the expected welfare over the collection of all convex cost functions. 
The worst-case cost function must be linear because the total cost under a linear cost function is highest among all convex cost functions with the same marginal cost at the same equilibrium precision of private information. 
Based on this observation, we show that the worst-case welfare necessarily increases with public information if and only if 
the volatility's coefficient is positive, in which case full disclosure is robust. 
This is true even when there exists a cost function such that 
the welfare can decrease with public information. 

One example is a Cournot game \citep{vives1988}, where a measure of material benefit is the total profit. Full disclosure is robust, whereas no disclosure can be optimal under a strictly convex cost. 
Another example is a beauty contest game \citep{morrisshin2002}, where a measure of material benefit is the negative of the mean squared error of an action from the state. 
The worst-case welfare necessarily increases with public information, whereas the welfare can decrease under a strictly convex cost or in the case of exogenous private information \citep{morrisshin2002}. 
The detrimental effect of public information in the latter case prompted the debate on central bank transparency and has been challenged by many papers \citep{svensson2006,angeletospavan2004, hellwig2005, cornandheinemann2008, colombofemminis2008}. 
Our result contributes to this debate by proving a new rationale for central bank transparency: more precise public information is beneficial to the worst-case welfare.




This paper is organized as follows. 
Section \ref{the model} introduces the model and discusses equilibria in subgames.  
We identify optimal disclosure in Section \ref{Information disclosure with known information costs} and robust disclosure in Section \ref{Robust disclosure section}. 
Section \ref{Cournot} is devoted to applications to a Cournot game and a beauty contest game. 
The last section concludes the paper. 
All proofs are relegated to the appendices. 

\subsection{Related literature}\label{Related literature}

This paper is related to two strands of the literature on Bayesian persuasion and information design.\footnote{See survey papers by \citet{morrisbergemann2017} and \citet{kamenica2019} among others.} 
The first incorporates costly information acquisition by a receiver into the Bayesian persuasion framework \citep{kamenicagentzkow2011}. 
\citet{lipnowskietal2020a} and \citet{bloedelsegal2020} consider a rationally inattentive receiver who can learn less information than what a sender provides, with substantially different modeling assumptions.\footnote{In \citet{bloedelsegal2020}, a receiver learns about the state by paying costly attention to the sender's signals. In \citet{lipnowskietal2020a}, a receiver acquires information about the state itself, but it is less than the information provided by the sender in the sense of Blackwell, followed by \citet{wei2021} and \citet{lipnowskietal2020a}.}
On the other hand, 
\citet{nizzottoetal2020} and \citet{matyskovamontes2021} consider a receiver who can acquire additional information after receiving a sender's signal, which is closer in spirit to our model. 
One of their findings is that additional information acquisition can be detrimental to the sender and the receiver. 
In the aforementioned papers, 
the sender's payoff is independent of the receiver's information cost.
In our paper, the sender's concern is social welfare, including the information cost.

The second strand of the literature incorporates the sender's ambiguity about the receiver's sources of information. 
In \citet{kosterina2021}, the sender and receiver have different priors, and the receiver's prior is unknown to the sender. 
In \citet{huweng2021} and \citet{dworczakpavan2020}, they have a common prior, but the receiver may exogenously observe additional private information that is unknown to the sender. 
\citet{huweng2021} characterize the worst-case optimal signals under various degrees of ambiguity. 
\citet{dworczakpavan2020} focus on the case of full ambiguity and introduce a novel lexicographic approach to a robust solution of Bayesian persuasion. 
The sender first identifies all signals that are worst-case optimal. 
From these signals, the sender chooses the one that maximizes the sender's payoff under some conjecture (e.g., there is no private information). 
In the papers mentioned above, the receiver's sources of information are assumed to be exogenous. 
In our paper, the sender faces ambiguity about the receivers' costly information acquisition. Thus, not only the precision but also the cost of private information is unknown to the sender. 

While the models in the above papers are based on the Bayesian persuasion framework, our model is an extension of a symmetric LQG game\footnote{Its origin goes back to the seminal paper by \citet{radner1962}. See \citet{vives2008}.} with a continuum of agents  \citep{angeletospavan2007}. \citet{morrisbergemann2013} characterize the set of all Bayes correlated equilibria in this game and advocate a problem of finding optimal ones.\footnote{\citet{ui2020} considers information design in general LQG games by formulating it as semidefinite programming.}  
\citet{uiyoshizawa2015} solve the problem by adopting a quadratic objective function. 
They also study optimal disclosure of public information in the case of exogenous private information. 
On the other hand, the present paper considers the case of endogenous private information by incorporating the information provision stage into the model of \citet{colombofemminis2014}, which integrates the symmetric LQG game and information acquisition with a convex information cost. 
A Cournot game with a linear cost \citep{lietal1987,vives1988} and a beauty contest game with a linear or convex cost \citep{colombofemminis2008, ui2014} are special cases of the model in \citet{colombofemminis2014}. 
\citet{colombofemminis2014} compare the equilibrium precision of private information with the socially optimal one, where the aggregate net payoff is a measure of welfare. 
They show that the social value of public information is less than that in the case of exogenous private information if and only if the equilibrium precision of private information is inefficiently low. 
They also provide a sufficient condition for the social value to be positive. 
In contrast, the present study gives a necessary and sufficient condition and obtains optimal and robust disclosure of public information.\footnote{Other models of information acquisition in symmetric LQG games are studied by \citet{mackowiakwiederholt2009}, \citet{hellwigveldkamp2009}, \citet{myattwallace2012,myattwallace2015}, \citet{denti2020}, \citet{rigos2020}, and \citet{hebertlao2021}, among others.}\footnote{\citet{myattwallace2019} and \citet{leister2020} study information acquisition in network LQG games.}

\section{The model}\label{the model}

There are a policymaker and a continuum of agents indexed by $i\in [0,1]$. 
The policymaker knows a random variable $\theta$ parametrizing the underlying state and provides agents with public information about $\theta$. 
At the same time, each agent endogenously acquires private information about $\theta$. 
More specifically, we consider the following three-period setting. In period 1, the policymaker chooses the precision of public information. 
In period 2, each agent chooses the precision of private information given that of public information. 
In period 3, each agent observes private and public signals and chooses an action in a symmetric linear quadratic Gaussian (LQG) game.

Let $y=\theta+\varepsilon_y$ and $x_i=\theta+\varepsilon_i$ 
be public and private signals of agent $i$, respectively, where $\varepsilon_y$, $\varepsilon_i$, and $\theta$ are independently and normally distributed with 
\[
\E[ \theta]=\bar\theta,\ 
\E[\varepsilon_y] =
\E [\varepsilon_i]=0,\ \mathrm{var}[ \theta]=\tau_\theta^{-1},\  
\mathrm{var}[\varepsilon_y]=\tau_y^{-1},\ 
\mathrm{var}[\varepsilon_i]=\tau_i^{-1}. 
\]
We refer to $\tau_y$ and $\tau_i$ as the precision of public information and that of private information, respectively. 
The policymaker chooses $\tau_y$ in period 1 at no cost. 
 Agent $i$ chooses $\tau_i$ in period 2 at a cost of $C(\tau_i)$. 
We assume that $C(\tau_i)$ is a strictly increasing convex function: $C(0)=0$, $C'(\tau_i)>0$, and $C''(\tau_i)\geq 0$.  
We denote the elasticity of marginal cost by $\rho(\tau_i)\equiv\tau_iC''(\tau_i)/C'(\tau_i)\geq 0$, which will play an essential role in our analysis.  
 


In period 3, agent $i$ chooses a real number $a_i \in \mathbb{R}$ as his action. 
We write $a=(a_i)_{i\in [0,1]}$ and $a_{-i}=(a_j)_{j\neq i}$. 
Agent $i$'s payoff to $a$ is symmetric and quadratic in $a$ and $\theta$: 
\begin{align}
u_i({a}, \theta)=&-a_i^2+2\alpha a_i \int_0^1 a_jdj+2\beta \theta a_i+h(a_{-i},\theta), \label{payoff function cont}
\end{align}
where $\alpha, \beta \in \mathbb{R}$ are constant and $h(a_{-i},\theta)$ is a measurable function. 
The best response is the interim expected value of a linear combination of the average action and the state:
\begin{equation}
\E\Big[\alpha \int a_jdj+\beta \theta\Big| x_i,y\Big].\label{linear best response}	
\end{equation}
This game exhibits strategic complementarity if $\alpha>0$ and strategic substitutability if $\alpha<0$. 
We assume $\alpha<1$, which guarantees the existence and uniqueness of a symmetric equilibrium in the last-period subgame when $\tau_i=\tau_j$ for all $i\neq j$. We also assume $\beta>0$ without loss of generality. 
Agent $i$ chooses $\tau_i$ and $a_i$ in periods 2 and 3 to maximize the expected value of his payoff function $u_i(a,\theta)-C(\tau_i)$. 



A welfare function is given by a material benefit minus the aggregate information cost 
\begin{equation}
v(a,\theta)-\int C(\tau_j)dj, \notag\label{bp's payoff1}
\end{equation}
where $v(a,\theta)$ is symmetric and quadratic in $a$ and $\theta$; that is,  
\begin{equation}
v(a,\theta)=c_1\int_0^1 a_j^2 dj+c_2\left(\int_0^1 a_j dj\right)^2+c_3\theta\int_0^1 a_j dj+c_4\int_0^1 a_j dj+c_5, \label{v function}
\end{equation}
where $c_1, c_2, c_3, c_4, c_5\in \mathbb{R}$ are constant.  
A typical case is the aggregate payoff 
(when $h(a_{-i},\theta)$ is a quadratic function), as considered by \citet{colombofemminis2014}. 
The policymaker chooses $\tau_y$ to maximize the expected welfare in period~1.

\subsection*{The subgames}

The last-period subgame has the following symmetric equilibrium \citep{angeletospavan2007}.\footnote{See \citet{uiyoshizawa2012} for the relationship between this result and the result of \citet{radner1962}.}

\begin{lemma}\label{AP lemma}
When $\tau_i=\tau_x$ for all $i$, the last-period subgame has a unique symmetric equilibrium. 
Agent $i$'s equilibrium strategy is 
$
\sigma_i(x_i,y)=b_x (x_i-\bar\theta)+b_y (y-\bar\theta)+{\beta\bar\theta}/({1-\alpha})$,
where
\begin{equation}
b_x=\frac{\beta}{(1-\alpha ) \tau_x+\tau_y+\tau_{\theta }}   \cdot\tau_x,\ 
b_y=\frac{\beta}{(1-\alpha ) \tau_x+\tau_y+\tau_{\theta }}\cdot\frac{  \tau_y}{  1-\alpha }.
\notag\label{eq coefficient}
\end{equation}
\end{lemma}


In a symmetric equilibrium of the second-period subgame, 
the precision of private information is determined by the following first-order condition \citep{colombofemminis2014}.\footnote{Using the envelope theorem, we can ignore the change in $br_i(\tau_i)$ in evaluating the left-hand side of \eqref{marginal benefit}. Thus, the marginal benefit is reduced to the partial derivative of $\E_{\tau_i}[u_i(\sigma, \theta)]$ with respect to $\tau_i$ holding fixed the equilibrium strategy profile $\sigma$ because $br_i(\tau_x)=\sigma_i$.
When we regard $u_i(\sigma, \theta)$ as a quadratic function of $\varepsilon_i$, the expected value of the linear term equals zero. Thus, the marginal benefit equals 
$\frac{d}{d \tau_i}\E_{\tau_i}[u_i(\sigma, \theta)]\Big|_{\tau_i=\tau_x}
=-\frac{d}{d \tau_i}\E_{\tau_i}[\sigma_i^2]\Big|_{\tau_i=\tau_x}
={b_x^2}/{\tau_x^2}$.
} 
  
\begin{lemma}\label{CF lemma}
Assume that all the opponents choose the same precision $\tau_j=\tau_x$ for all $j\neq i$ in the second-period subgame and follow the unique symmetric equilibrium strategy $\sigma_j$ in the last-period subgame.  
Agent $i$'s marginal benefit of choosing $\tau_i$ evaluated at $\tau_i=\tau_x$ is 
\begin{equation}
\left.\frac{d}{d \tau_i}\E_{\tau_i}[u_i((br_i(\tau_i),\sigma_{-i}), \theta)]\right|_{\tau_i=\tau_x}
=\frac{b_x^2}{\tau_x^2}=\frac{\beta^2 }{\left((1-\alpha)\tau_x+\tau_y+\tau_\theta\right)^2},
\label{marginal benefit}
\end{equation}
where $br_i(\tau_i)\equiv\E_{\tau_i}[\alpha\int \sigma_jdj+\beta\theta|x_i,y]$ is the best response to $\sigma_{-i}=(\sigma_{j})_{j\neq i}$ and $\E_{\tau_i}$ is the expectation operator when player $i$'s precision is $\tau_i$. 
\end{lemma}

By Lemma \ref{CF lemma}, 
the first-order condition for the precision in a symmetric equilibrium is  
\begin{equation}
\frac{\beta^2 }{\left((1-\alpha)\tau_x+\tau_y+\tau_\theta\right)^2}=
C'(\tau_x). \label{FOC1}
\end{equation}
Note that the marginal benefit is strictly decreasing in $\tau_x$ and $\tau_y$, whereas the marginal cost is increasing in $\tau_x$.
Thus, the equilibrium precision is the unique value of $\tau_x$ solving \eqref{FOC1} if $C'(0)< \beta^2/(\tau_y+\tau_\theta)^2$ and zero if $C'(0)\geq \beta^2/(\tau_y+\tau_\theta)^2$. 
We denote the equilibrium precision by $\phi(\tau_y)$.

If $C$ is linear with a marginal cost $c>0$, i.e., $C(\tau_x)=c\tau_x$, then 
$\phi(\tau_y)$ is given by\footnote{This expression is obtained by \citet{lietal1987} and \citet{vives1988} for Cournot games and \citet{colombofemminis2008} for beauty contest games.} 
\begin{equation}
\phi^0_c(\tau_y)\equiv
\begin{cases}
 \left({\beta}/{\sqrt{c}}-\tau_y-\tau_\theta\right)/(1-\alpha) &\text{ if }c< \beta^2/(\tau_y+\tau_\theta)^2,\\
0 &\text{ if }c\geq \beta^2/(\tau_y+\tau_\theta)^2.
\end{cases}\label{linear cost}
\end{equation}
If $C$ is nonlinear, $\phi(\tau_y)$ does not have a closed-form expression, so we will rely on its inverse $\phi^{-1}(\tau_x)$ in our analysis: for $\tau_x>0$,
\begin{equation}
\phi^{-1}(\tau_x)=-(1-\alpha)\tau_x-\tau_\theta+{\beta}/{\sqrt{C'(\tau_x)}}.\label{phiinv}
\end{equation}

Because $\phi'=(d\phi^{-1}/d\tau_x)^{-1}<0$, an increase in the precision of public information results in a decrease in the precision of private information, as shown by \citet{colombofemminis2014}. 
This effect is referred to as the crowding-out effect of public information on private information,\footnote{A similar effect is also found in \citet{colombofemminis2008}, \citet{wong2008}, \citet{hellwigveldkamp2009}, and \citet{myattwallace2012}, among others. On the other hand, \citet{caidong2021} discuss the crowd-in effect of public information. } which is illustrated in Figure \ref{figure}.    
It is known that the crowding-out effect is largest when $C$ is linear, as stated in the following lemma \citep{ui2014}. 
\begin{lemma}\label{mitigate lemma}
Let $\phi(\tau_y)$ be the equilibrium precision under a strictly convex cost function. 
If $\phi(\tau_y)=\phi^0_{c}(\tau_y)>0$, then
$d\phi^0_c(\tau_y)/d\tau_y<d\phi(\tau_y)/d\tau_y<0$, where $\phi^0_{c}(\tau_y)$ is the equilibrium precision under a linear cost function with a marginal cost $c>0$. 
\end{lemma}

\begin{figure}
\centering
\includegraphics[width=7cm, bb=0 0 450 279]{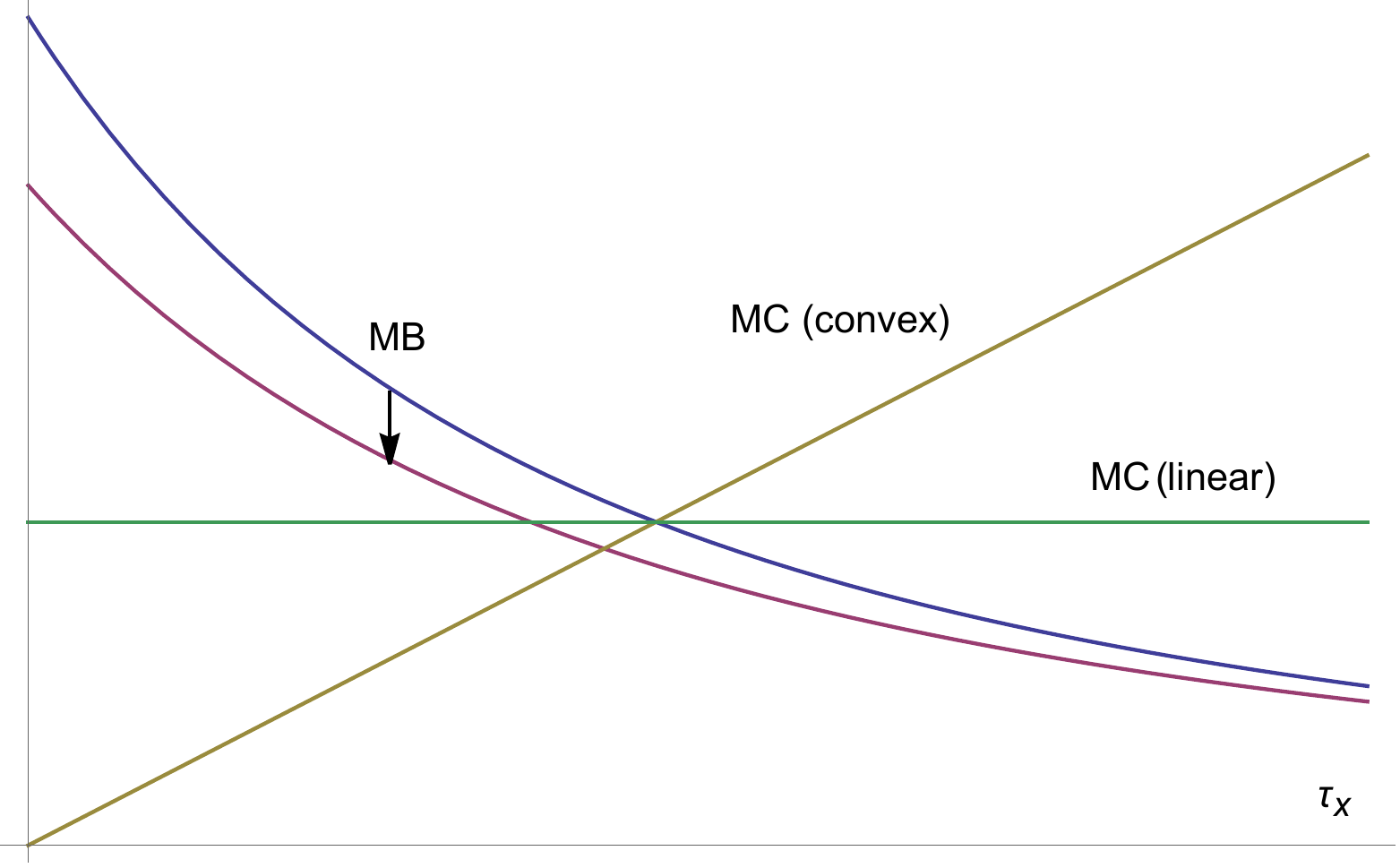}
\caption{The marginal benefit curve and the marginal cost curve. The horizontal axis is the $\tau_x$-axis. 
An increase in the precision of public information shifts the marginal benefit curve down, 
by which the equilibrium precision decreases.}
\label{figure}
\end{figure}

\subsection*{The whole game}

Given the equilibrium of the second-period subgame, 
the policymaker chooses the precision of public information to maximize the expected welfare in period 1, which is calculated as 
\begin{align}
\E[v(\sigma,\theta)]-C(\tau_x)=c_1\mathrm{var}[\sigma_i]+c_2\mathrm{cov}[\sigma_i,\sigma_j]+c_3\mathrm{cov}[\sigma_i,\theta]-C(\tau_x) +\text{const.}\label{expected welfare 1}
\end{align}
by \eqref{v function} and Lemma \ref{AP lemma}, where $\tau_x=\phi(\tau_y)$.

We rewrite \eqref{expected welfare 1} using volatility and dispersion of actions. 
The volatility is the variance of the average action $\int \sigma_jdj$, 
and the dispersion is the variance of the idiosyncratic difference $\sigma_i-\int \sigma_jdj$. 
As shown by \citet{morrisbergemann2013}, the volatility equals the covariance of actions, $\mathrm{cov}[\sigma_i,\sigma_j]$, and the dispersion equals the difference between the variance and the covariance of actions, $\mathrm{var}[{\sigma_i}]-\mathrm{cov}[\sigma_i,\sigma_j]$.

Note that the first three terms in \eqref{expected welfare 1} are linearly dependent because 
$\mathrm{var}[\sigma_i]=\alpha \mathrm{cov}[\sigma_i,\sigma_j]+\beta \mathrm{cov}[\sigma_i,\theta]$ by \eqref{linear best response}.\footnote{We can obtain this by mulitplying both sides of \eqref{linear best response} by $a_i$ and taking the expectation.}  
Hence, we can rewrite \eqref{expected welfare 1} as an affine function of  $\mathrm{var}[\sigma_i]$ and $\mathrm{cov}[\sigma_i,\sigma_j]$, which is also an affine function of the volatility and dispersion \citep{uiyoshizawa2015}. 
\begin{lemma}\label{lemma 2}
When the precision of private and public information is $(\tau_x,\tau_y)$ and the agents follow the unique symmetric equilibrium in the last-period subgame,  
the expected welfare equals 
\begin{align}
W(\tau_x,\tau_y)
\equiv \zeta D(\tau_x,\tau_y)+\eta V(\tau_x,\tau_y)-C(\tau_x) \label{welfare simple}
\end{align}
plus a constant, 
where $\zeta=c_1+c_3/\beta$, $\eta=c_1+c_2+(1-\alpha)c_3/\beta$, and 
\begin{align}
V(\tau_x,\tau_y)& 
=\frac{b_y^2}{\tau_y}+\frac{(b_x+b_y)^2}{\tau_\theta}
=\frac{\beta^2   \left((1-\alpha )^2 \tau_x^2+2 (1-\alpha ) \tau_x \tau_y+\tau_y \left(\tau_{\theta }+\tau_y\right)\right)}{(1-\alpha )^2 \tau_{\theta } \left((1-\alpha ) \tau_x+\tau_y+\tau_{\theta }\right)^2},
\label{CV}\\
D(\tau_x,\tau_y)&
=\frac{b_x^2}{\tau_x}=\frac{\beta^2\tau_x}{\left((1-\alpha ) \tau_x+\tau_y+\tau_{\theta }\right)^2} 
\label{IV}
\end{align}
are the volatility and the dispersion, respectively.
\end{lemma}

The policymaker maximizes $W(\phi(\tau_y),\tau_y)$ with respect to $\tau_y$ in equilibrium. 
Thus, optimal disclosure is determined by the cost function $C$ and the coefficients of volatility and dispersion $\eta$ and $\zeta$ (i.e., constants independent of $\tau_x$ and $\tau_y$). 

Using a similar representation, 
\citet{uiyoshizawa2015} 
characterize optimal disclosure under exogenous private information in terms of the ratio of $\eta$ to $\zeta$.   
As will be demonstrated in the next section, the representation \eqref{welfare simple} plays a more critical role in the case of endogenous private information.

\section{Information disclosure with known information costs}\label{Information disclosure with known information costs}

In Section \ref{The welfare effect of public information},  
we evaluate the sign of the derivative of the expected welfare
\begin{equation}
\frac{d W(\phi(\tau_y),\tau_y)}{d \tau_y}=\zeta \frac{d D(\phi(\tau_y),\tau_y)}{d \tau_y}+\eta \frac{d V(\phi(\tau_y),\tau_y)}{d \tau_y}-\frac{d C(\phi(\tau_y))}{d\tau_y} 
\label{dwdtauy}
\end{equation}
in terms of $\zeta$, $\eta$, and the elasticity of marginal cost $\rho$. 
This derivative is referred to as the social value of public information. 
In Section \ref{The optimal disclosure of public information},  
we obtain the optimal precision of public information that maximizes $W(\phi(\tau_y),\tau_y)$ assuming an isoelastic cost function.

\subsection{The social value of public information}\label{The welfare effect of public information}

Because $\phi$ does not have a closed-form expression, 
 \eqref{dwdtauy} is not easy to calculate. 
Thus, we divide \eqref{dwdtauy} by $|dD/d\tau_y|$ and evaluate its sign instead. 
This value is denoted by $MW_D(\tau_y)$ and referred to as the marginal welfare with respect to dispersion (MWD): 
\[
MW_D(\tau_y)\equiv
\frac{d{W(\phi(\tau_y),\tau_y)}}{d\tau_y}/
\bigg|\frac{d D(\phi(\tau_y),\tau_y)}{d \tau_y}\bigg|.
\]
Note that $MW_D(\tau_y)$ and $dW(\phi(\tau_y),\tau_y)/d\tau_y$ have the same sign.   

To understand what $MW_D(\tau_y)$ measures, observe that more precise public information decreases the dispersion. 
This is because the dispersion equals the difference between the variance and the covariance of actions, and more precise public information brings the variance and the covariance closer. 
Therefore, when the precision of public information increases by $|dD/d\tau_y|^{-1}$, the dispersion decreases by one, and the welfare increases by $MW_D(\tau_y)$.
In fact, we have
\begin{gather}
	D(\tau_x,\tau_y)=\tau_x C'(\tau_x), \label{FOC2}\\
\frac{d D(\phi(\tau_y),\tau_y)}{d \tau_y}
=(1+ \rho)C'(\tau_x)\phi'(\tau_y)<0\label{dD}
\end{gather}
for $\tau_x=\phi(\tau_y)>0$  
by \eqref{FOC1} and \eqref{IV}, 
where $\rho=\rho(\tau_x)$.

If either $C$ is linear or private information is exogenous, the MWD is directly calculated.
\begin{lemma}\label{MWD special cases}
Suppose that $\phi(\tau_y)>0$. 
When $C$ is linear, the MWD is
a constant 
\begin{equation}
{MW_D}^0\equiv \eta/(1-\alpha) -\zeta+1.\notag\label{MWD0}
\end{equation}
When $\tau_x=\phi(\tau_y)$ is fixed, the MWD is
\begin{equation}
{MW_D}^*(\tau_y)\equiv 
-\frac{\partial W(\tau_x,\tau_y)}{\partial \tau_y}/\frac{\partial D(\tau_x,\tau_y)}{\partial \tau_y}\bigg|_{\tau_x=\phi(\tau_y)}
=\eta \frac{3 (1-\alpha) \phi(\tau_y)+\tau_y+\tau_\theta}{2 (1-\alpha)^2 \phi(\tau_y)}-\zeta.
\label{MWD*}
\end{equation}
\end{lemma}

Moreover, when $C$ is nonlinear, the MWD is represented as the weighted average of those in the above special cases, ${MW_D}^0$ and ${MW_D}^*(\tau_y)$, 
where the relative weights are one and the elasticity of marginal cost, respectively.
Because the crowding-out effect is largest when $C$ is linear (see Lemma \ref{mitigate lemma}),  
the MWD equals the weighted average of those under the largest crowding-out effect and no crowding-out effect. 

\begin{proposition}\label{main proposition 1}
Suppose that $\phi(\tau_y)>0$. 
Then, 
\begin{align}
MW_D(\tau_y)
&=\frac{{MW_D}^0+\rho {MW_D}^*(\tau_y)}{1+\rho}, \notag\label{MWD calc}
\end{align}
where $\rho=\rho(\phi(\tau_y))$.
\end{proposition}

In summary, the social value of public information has the same sign as ${MW_D}^0$ if the crowding-out effect is sufficiently large (with small $\rho$) and ${MW_D}^*(\tau_y)$ if the crowding-out effect is sufficiently small (with large $\rho$). 
In particular, when ${MW_D}^0>0>{MW_D}^*(\tau_y)$,  
the sign is positive if the crowding-out effect is large enough, while negative under exogenous private information. 
In other words, the crowding-out effect can turn the social value of public information from negative to positive if and only if ${MW_D}^0>0>{MW_D}^*(\tau_y)$. 
Similarly, the crowding-out effect can turn the social value of public information from positive to negative if and only if ${MW_D}^0<0<{MW_D}^*(\tau_y)$.

We restate Proposition \ref{main proposition 1} using $(\zeta,\eta)$ and 
\begin{equation}
MV_D(\tau_y)\equiv\frac{d V(\phi(\tau_y),\tau_y)}{d \tau_y}/\bigg|\frac{d D(\phi(\tau_y),\tau_y)}{d \tau_y}\bigg|,\notag
\end{equation}
which is referred to as the marginal volatility with respect to dispersion (MVD).

\begin{corollary}\label{main corollary 1}
Suppose that $\phi(\tau_y)>0$. 
Then, 
\begin{align}
MW_D(\tau_y)
=\eta {MV_D}(\tau_y)-\zeta+\frac{1}{1+\rho}, \label{MWD calc2}
\end{align}
where $\rho=\rho(\phi(\tau_y))$.
The MVD is given by 
\begin{equation}
MV_D(\tau_y)=\frac{{MV_D}^0+\rho\, {MV_D}^*(\tau_y)}{1+\rho}>0,\label{MRTgeneral}
\end{equation}
where 
${MV_D}^{0}\equiv{1}/{(1-\alpha)}$ 
is the MVD when $C$ is linear, and 
\begin{align}
{MV_D}^*(\tau_y)
\equiv \frac{3 (1-\alpha) \phi(\tau_y)+\tau_y+\tau_\theta}{2 (1-\alpha)^2 \phi(\tau_y)}
>
\frac{3 }{2 (1-\alpha)}>{MV_D}^{0}
\label{MRTex}
\end{align}
is the MVD when $\tau_x=\phi(\tau_y)$ is fixed.
\end{corollary}

According to Corollary \ref{main corollary 1}, a change in the welfare is broken down into the following: when the precision of public information increases by $|dD/d\tau_y|^{-1}$, the volatility term increases by $\eta MV_D(\tau_y)$, the dispersion term decreases by $\zeta$, and the information cost decreases by $1/(1+\rho)$. 
Note that $MV_D(\tau_y)\geq {MV_D}^0>0$; that is, the MVD is strictly positive. 
This is because more precise public information increases the covariance of actions, which equals the volatility. 
Consequently, more precise public information increases the welfare if $\eta$ is sufficiently large compared to $\zeta$ and decreases the welfare if $\eta$ is sufficiently small compared to $\zeta$. 

When $C$ is linear, the MVD is smallest, and the cost reduction $1/(1+\rho)$ is largest.  
In other words, the crowding-out effect induces a smaller increase in the volatility and a larger decrease in the cost. 
Thus, when $\eta>0$,  
the crowding-out effect can turn the sign of the social value of public information for either of two reasons: from negative to positive 
due to a larger decrease in the cost, and from positive to negative 
due to a smaller increase in the volatility. 



As another corollary of Proposition \ref{main proposition 1}, we provide the ranges of $\eta$ guaranteeing ${MW_D}(\tau_y)>0$ and guaranteeing ${MW_D}(\tau_y)<0$, respectively, which will be used in the study of optimal disclosure. 

\begin{corollary}\label{main corollary 2}
Let 
\[
\underline\eta(\zeta,\rho)\equiv \frac{2(1-\alpha)\left((1+\rho)\zeta-1\right)}{3\rho+2}.
\]
For $\tau_y\geq 0$ with $\phi(\tau_y)>0$, 
 ${MW_D}(\tau_y)>0$ if $\eta>\max\{\underline\eta(\zeta,\rho),0\}$, and 
 ${MW_D}(\tau_y)<0$ if $\eta<\min\{\underline\eta(\zeta,\rho),0\}$, where $\rho=\rho(\phi(\tau_y))$. 
\end{corollary}



\subsection{The optimal disclosure of public information}\label{The optimal disclosure of public information}

We consider the optimal precision of public information in the following sense, which constitutes a subgame perfect equilibrium of our model. 
\begin{definition}
The precision of public information $\tau_y^*\in\mathbb{R}_+\cup\{\infty\}$ is optimal if 
\[
W(\phi(\tau_y^*),\tau_y^*)=\sup_{\tau_y}W(\phi(\tau_y),\tau_y). 
\]	
We say that full disclosure is optimal if $\tau_y^*=\infty$, no disclosure is optimal if $\tau_y^*=0$, and partial disclosure is optimal if $\tau_y^*>0$ is finite. 
\end{definition}

We can obtain the optimal precision by using Proposition \ref{main proposition 1} and identifying all local optima. 
To simplify the discussion, we focus on the case of an isoelastic cost function,  
\begin{equation}
	C(\tau_x)=c\tau_x^{\lambda+1}/(\lambda+1),\label{isoelastic cost function}
\end{equation} 
where $\lambda=\rho(\tau_x)\geq 0$ and $c>0$ are constant. 
Under this assumption, the cost in the second-period subgame is proportional to the dispersion by \eqref{FOC2}: 
\begin{equation}
C(\phi(\tau_y))=D(\phi(\tau_y),\tau_y)/(\lambda+1).\label{cost and dispersion}
\end{equation}

Using Corollary \ref{main corollary 1}, 
we can rewrite $MW_D(\tau_y)$ as a linear function of ${MV_D}^*(\tau_y)$,
\[
MW_D(\tau_y)=\frac{\eta \lambda }{\lambda+1}{MV_D}^*(\tau_y)-\zeta+\frac{\eta 
/(1-\alpha)+1}{\lambda+1},
\]
which depends on $\tau_y$ only through ${MV_D}^*(\tau_y)$. 
Note that ${MV_D}^*(\tau_y)$ is increasing in $\tau_y$ and goes to infinity as $\tau_y$ goes to infinity. 
Thus, if $\eta>0$, then $MW_D(\tau_y)$ is increasing in $\tau_y$ and $MW_D(\tau_y)>0$ for sufficiently large $\tau_y$. 
This implies that full disclosure is optimal if $MW_D(0)>0$ (see Figure \ref{fig:case (i)}), which is the case if $\eta>\max\{\underline{\eta}(\zeta,\lambda),0\}$ by Corollary~\ref{main corollary 2}; otherwise, no disclosure can be optimal (see Figure \ref{fig:case (ii)}). 
Similarly, if $\eta<0$, then $MW_D(\tau_y)$ is decreasing in $\tau_y$ and $MW_D(\tau_y)<0$ for sufficiently large $\tau_y$. 
This implies that no disclosure is optimal if $MW_D(0)<0$ (see Figure \ref{fig:case (iii)}), which is the case if $\eta<\min\{\underline{\eta}(\zeta,\lambda),0\}$ by Corollary~\ref{main corollary 2}; otherwise, paritial disclosure is optimal (see Figure~\ref{fig:case (iv)}).

The following characterization of the optimal precision follows from the above discussion, which is illustrated in Figure \ref{figure 2}.

\begin{figure} 
  \centering
  \subfloat[Case (i)]{\label{fig:case (i)}
  \includegraphics[width=4cm, bb=0 0 360 229]{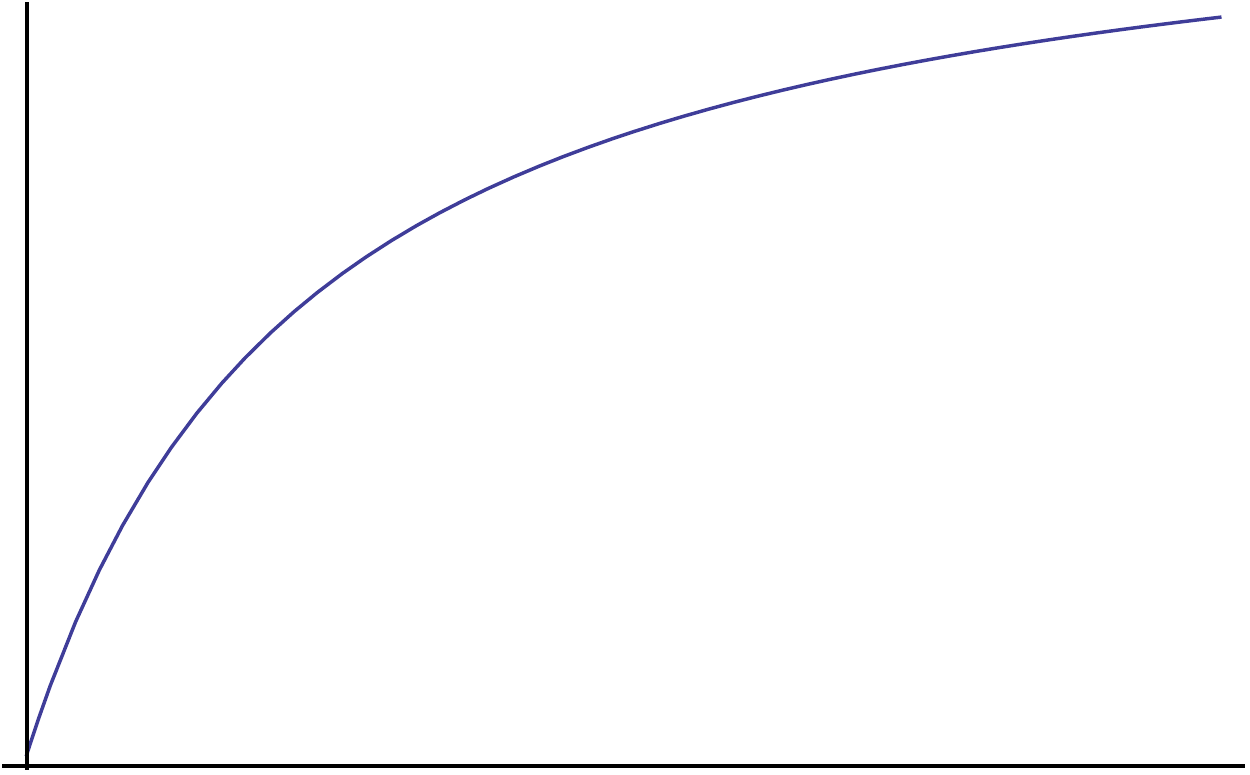}}\qquad \ 
  \subfloat[Case (ii)]{\label{fig:case (ii)}\includegraphics[width=4cm, bb=0 0 360 229]{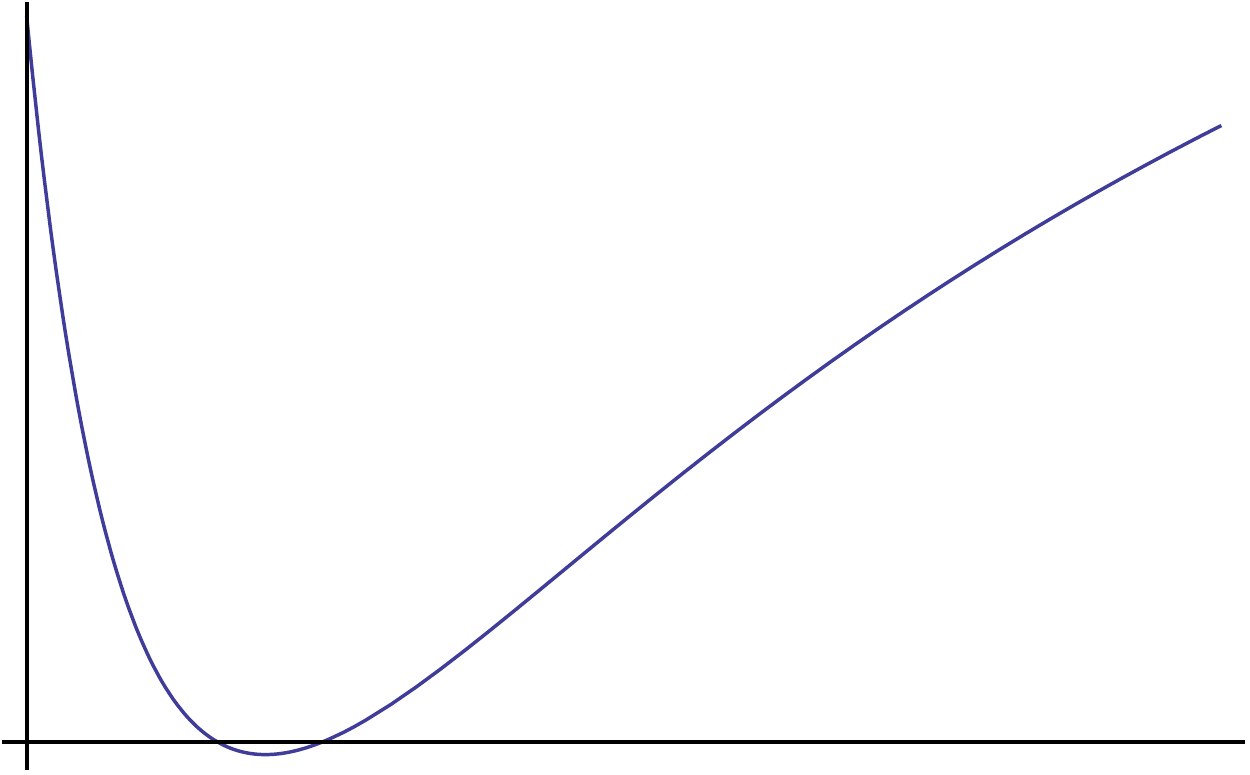}}\\ 
  \subfloat[Case (iii)]{\label{fig:case (iii)}\includegraphics[width=4cm, bb=0 0 360 229]{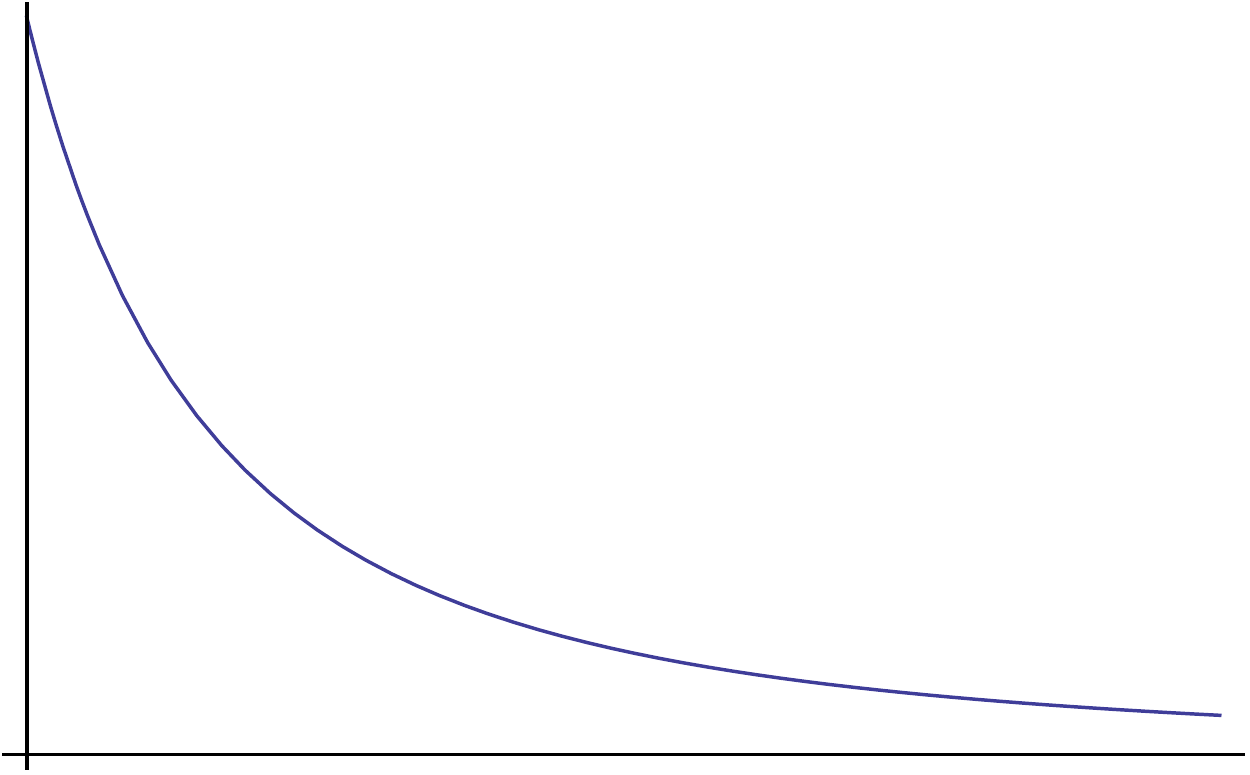}}\qquad \ 
  \subfloat[Case (iv)]{\label{fig:case (iv)}\includegraphics[width=4cm, bb=0 0 360 229]{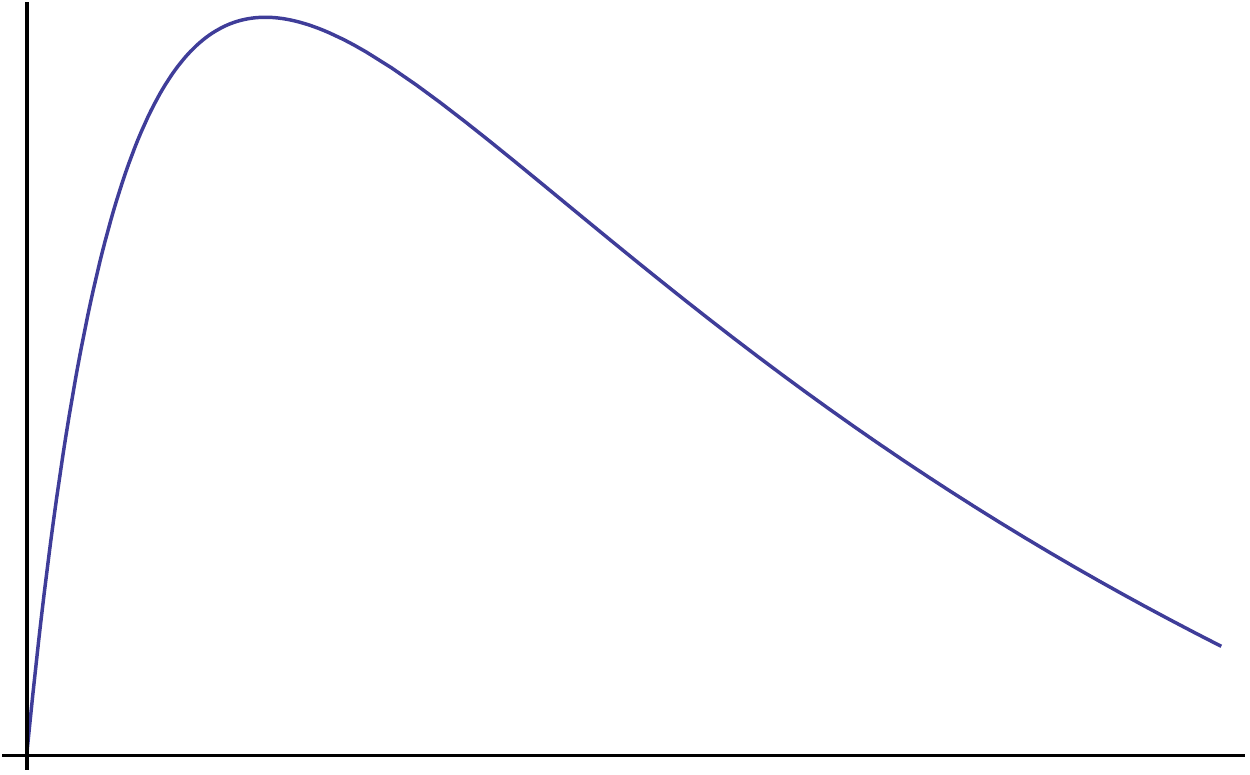}}
  \caption{The graph of the welfare as a function of the precision. } 
  \label{fig1}
\end{figure}

\begin{proposition}\label{optimal proposition}
Assume an isoelastic cost function \eqref{isoelastic cost function}.
Then, the following holds. 
\begin{enumerate}[\em (i)]
\item 
Suppose that $\eta> \max\{\underline\eta(\zeta,\lambda), 0\}$. 
Then,  
${d W(\phi(\tau_y),\tau_y)}/{d \tau_y}> 0$ for all $\tau_y$. 
Thus, the optimal precision is $\infty$. 

\item 
Suppose that $0<\eta< \underline\eta(\zeta,\lambda)$. 
Then,  
${d W(\phi(\tau_y),\tau_y)}/{d \tau_y}\lessgtr 0$ for $\tau_y\lessgtr\bar\tau_y$, where 
\begin{align}
\bar \tau_y&\equiv\phi^{-1}(\bar\tau_x)=\bar\tau_z-\tau_\theta,\notag\\ 
\bar\tau_z &\equiv
\begin{cases}{\beta}/{\sqrt{c}}&\text{ if }\lambda=0,	\\
\bar\tau_x{(1-\alpha ) (3 \lambda +2)(\underline\eta(\zeta,\lambda)-\eta  )}/{(\eta  \lambda )}&\text{ if }\lambda>0,	
\label{phiinv'}
\end{cases}
\\
\bar\tau_x &\equiv\left(\frac{\beta   \lambda \eta}{2 c^{1/2}\left(1-\alpha\right) \left(\left(1-\alpha\right) ((1+\lambda)\zeta -1)-(1+\lambda)\eta \right)}\right)^{2/(\lambda+2)}.\label{taux-star'}
\end{align} 
Thus, the optimal precision is $\infty$ if $\tau_\theta\geq \bar\tau_z$ and either $0$ or $\infty$ if $\tau_\theta< \bar\tau_z$.

\item 
Suppose that $\eta< \min\{\underline\eta(\zeta,\lambda), 0\}$. Then, 
${d W(\phi(\tau_y),\tau_y)}/{d \tau_y}< 0$ for all $\tau_y$. 
Thus, the optimal precision is $0$.

\item 
Suppose that $0>\eta> \underline\eta(\zeta,\lambda)$. Then, 
${d W(\phi(\tau_y),\tau_y)}/{d \tau_y}\gtrless 0$ for $\tau_y\lessgtr\bar\tau_y$. 
Thus, the optimal precision is $\max\{0,\bar\tau_y\}$.
\end{enumerate}
\end{proposition}

\begin{figure}
\centering
\includegraphics[width=8cm, bb=0 0 540 335]{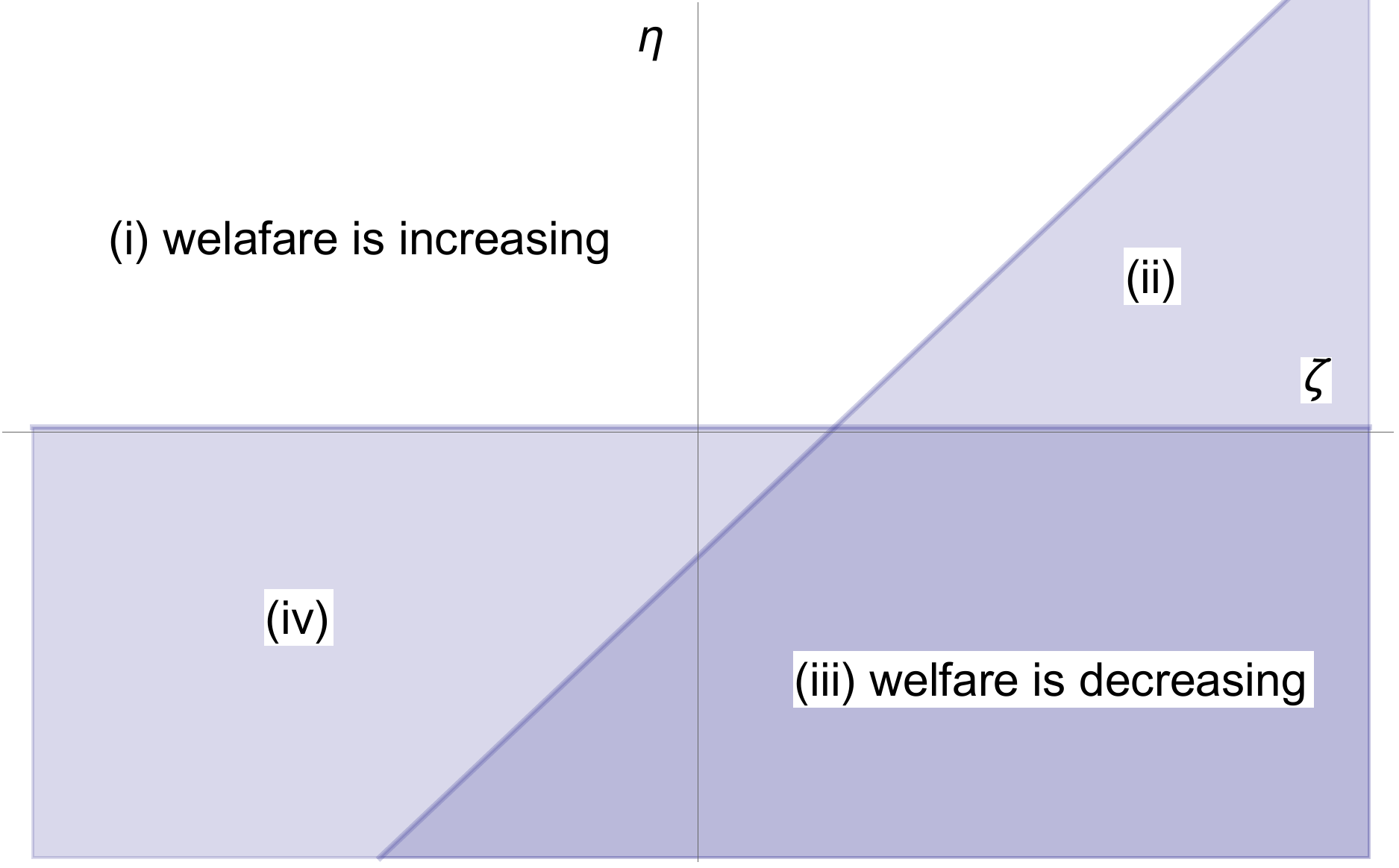}
\caption{The four cases on the $\zeta$-$\eta$ plane. 
The upward sloping line is the graph of $\eta=\underline\eta(\zeta,\lambda)$. Full disclosure is optimal in the region (i). Full or no disclosure is optimal in the region (ii). No disclosure is optimal in the region (iii). Partial or no disclosure is optimal in the region (iv). 
}
\label{figure 2}
\end{figure}

Note that, even in the case of (ii), 
full disclosure is optimal if $\bar\tau_y<0$. 
For example, this is the case 
if the cost function is linear and the marginal cost is sufficiently large. 
Indeed,  $\bar\tau_y={\beta}/{\sqrt{c}}-\tau_\theta<0$ for $c>\beta^2/\tau_\theta^2$ by \eqref{phiinv'} when $\lambda=0$.  
This observation is summarized in the following result and will be used in the next section.

\begin{corollary}\label{optimal proposition corollary}
Suppose that $\eta\neq 0$. Let $\lambda=0$ and $c>\beta^2/\tau_\theta^2$. Then, full disclosure is optimal if and only if $\eta>0$. 
\end{corollary}

We can also consider the gross expected welfare excluding the information cost. 
From \eqref{cost and dispersion}, the gross welfare is represented as 
\begin{align}
W^g(\phi(\tau_y),\tau_y)
&\equiv \zeta D(\phi(\tau_y),\tau_y)+\eta V(\phi(\tau_y),\tau_y)\notag \\
&=(\zeta+1/(\lambda+1)) D(\phi(\tau_y),\tau_y)+\eta V(\phi(\tau_y),\tau_y)-C(\phi(\tau_y)).
\label{gross welfare}\end{align}
Thus, we can use Proposition~\ref{optimal proposition} by replacing $\zeta$ with $\zeta+1/(\lambda+1)$. 

\section{Information disclosure with unknown information costs}
\label{Robust disclosure section}

Assume that the policymaker does not know what cost function the agents have.  
Even in this case, if the welfare increases with public information regardless of a cost function, 
the policymaker prefers more precise public information.
In Section \ref{regardless of costs}, we identify the class of welfare functions having this property. 
If a welfare function is not in this class, under what condition should the policymaker provide more precise public information? 
In Section \ref{robust disclosure}, we address this question by assuming that the policymaker evaluates the precision of public information in terms of the worst-case welfare.

\subsection{Positive (negative) social value regardless of a cost function}\label{regardless of costs}

We identify the class of welfare functions such that ${d W(\phi(\tau_y),\tau_y)}/{d \tau_y}> 0$ regardless of a cost function. 
According to Corollary \ref{main corollary 2},  
if $\eta>\max\{\underline\eta(\zeta,\rho),0\}$ for all $\rho\geq 0$, then ${d W(\phi(\tau_y),\tau_y)}/{d \tau_y}> 0$ for an arbitrary cost function with $\phi(\tau_y)>0$, which is valid for any $\tau_y$.
By elaborating on this argument, we obtain the following proposition.

\begin{proposition}\label{corollary main result 2}
The following three statements are equivalent. 
\begin{enumerate}[\em (i)]
	\item {There exists $\tau_y\geq 0$ such that, for any cost function with $\phi(\tau_y)>0$, ${d W(\phi(\tau_y),\tau_y)}/{d \tau_y}> 0$.} 
	\item For any cost function and any $\tau_y$ with $\phi(\tau_y)>0$, ${d W(\phi(\tau_y),\tau_y)}/{d \tau_y}> 0$. 
	\item $(\zeta,\eta)$ satisfies
\begin{equation}
\eta
\begin{cases}
\geq 0 &\text{ if }\zeta\leq 0,\\
\geq \underline{\eta}(\zeta,\infty)=2(1-\alpha)\zeta/3 &\text{ if }0< \zeta< 3,\\
>\underline{\eta}(\zeta,0)=(1-\alpha)(\zeta-1)&\text{ if }\zeta \geq 3.
\end{cases} \notag\label{robust welinc}
\end{equation}
\end{enumerate}
Similarly, the following three statements are equivalent. 
\begin{enumerate}[\em (i')]
	\item There exists $\tau_y\geq 0$ such that, for any cost function with $\phi(\tau_y)>0$, ${d W(\phi(\tau_y),\tau_y)}/{d \tau_y}< 0$. 
	\item For any cost function and any $\tau_y$ with $\phi(\tau_y)>0$, ${d W(\phi(\tau_y),\tau_y)}/{d \tau_y}< 0$.
	\item $(\zeta,\eta)$ satisfies 
\begin{equation}
\eta
\begin{cases}
\displaystyle
< \underline{\eta}(\zeta,0)=(1-\alpha)(\zeta-1) &\text{ if }\zeta \leq 1,\\
\leq 0 &\text{ if }\zeta > 1.\\
\end{cases}\notag\label{robust weldec}
\end{equation}
\end{enumerate}
\end{proposition}

Figure \ref{figure 4} illustrates the classes of welfare functions identified by Corollary \ref{corollary main result 2} on the $\zeta$-$\eta$ plane. 
Regardless of a cost function, the welfare increases with public information if $(\zeta,\eta)$ is in the upper-left area; the welfare decreases with public information if $(\zeta,\eta)$ is in the lower-right area. 
If $(\zeta,\eta)$ is not in these areas, there exists a cost function such that the welfare is not a monotone function.

\begin{figure}
\centering
\includegraphics[width=8cm, bb=0 0 540 335]{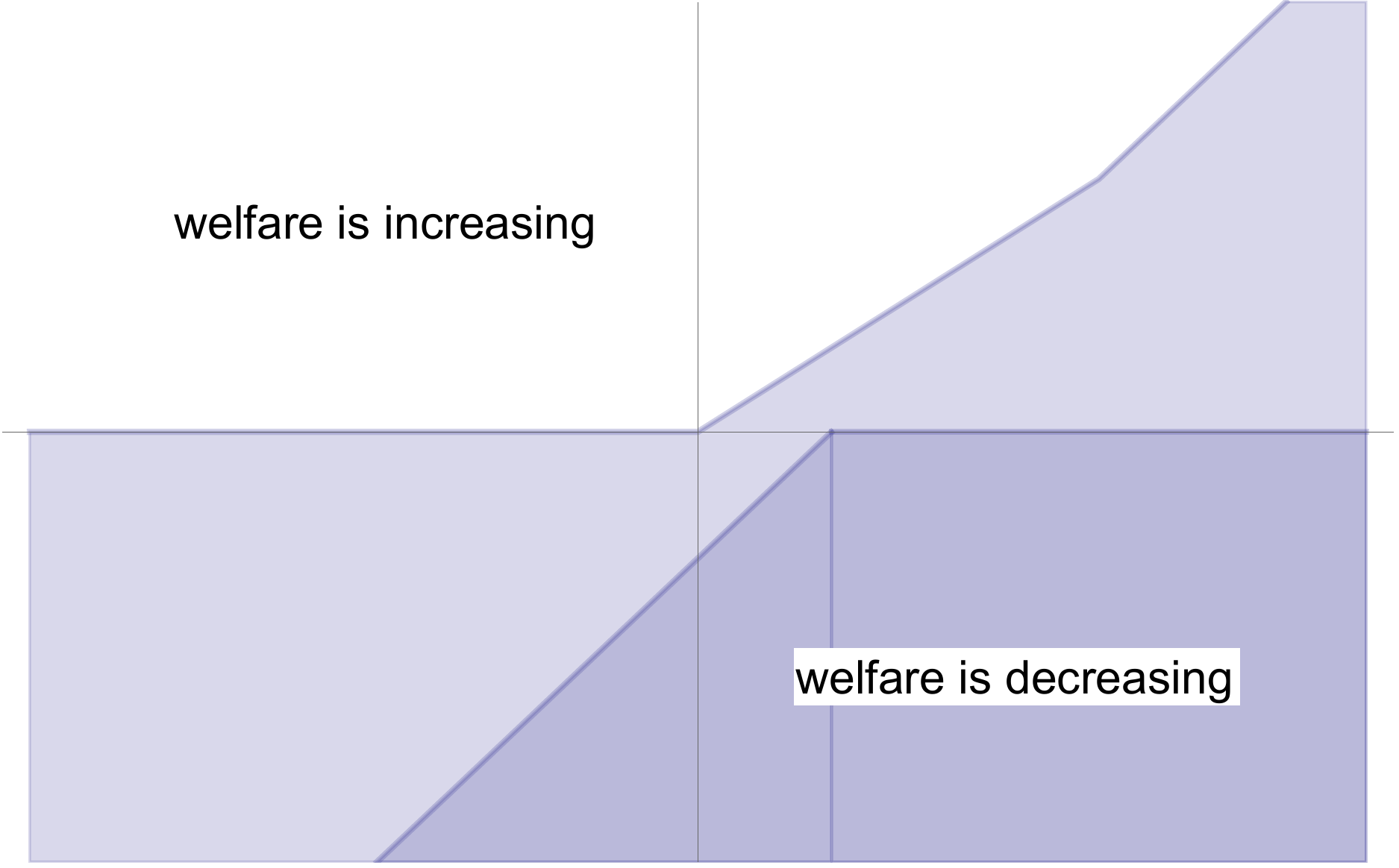}
\caption{ The welfare effects of public information under an arbitrarily cost function.}
\label{figure 4}
\end{figure}

\subsection{The robust disclosure of public information}\label{robust disclosure}

Consider a policymaker with a welfare function that does not satisfy either condition in Proposition~\ref{corollary main result 2}. 
Assume that the policymaker evaluates the precision of public information in terms of the worst-case welfare, which is the infimum of the expected welfare over the collection of all convex cost functions. We denote the collection by
\[
\Gamma\equiv\{C:\mathbb{R}_+\to\mathbb{R}_+\mid \text{$C(0)=0$, $C'(\tau_x)>0$, $C''(\tau_x)\geq 0$ for all $\tau_x\geq 0$}\}.
\]
For each $C\in \Gamma$, we write $\phi_C(\tau_y)$ for the precision of private information when 
the cost function is $C$ and the precision of public information is $\tau_y$. 
We also consider the subcollection of cost functions such that 
the precision of private information is less than or equal to $\kappa\in\mathbb{R}_{++}
\cup\{\infty\}$ for all $\tau_y$, which is given by  
\[
\Gamma_\kappa\equiv\{C\in\Gamma\mid \phi_C(\tau_y)\leq \kappa \text{ for all }\tau_y\}
=\{C\in\Gamma\mid C'(\kappa)\geq {\beta^2 }/{\left((1-\alpha)\kappa+\tau_\theta\right)^2}\}.
\]
The above equality follows from the first-order condition for the precision \eqref{FOC1}. 
Note that $\Gamma=\Gamma_\infty$ and $\Gamma_\kappa\supseteq \Gamma_{\kappa'}$ if $\kappa\geq\kappa'$.

For each $\tau_y$, let 
\begin{align}
F_\kappa(\tau_y)&\equiv 
\inf_{C\in \Gamma_\kappa}W(\phi_C(\tau_y),\tau_y)\notag\label{worst-case welfare:def}
\end{align}
be the infimum of the expected welfare over $\Gamma_\kappa$. 
We consider $F_\kappa(\tau_y)$ with $\kappa<\infty$ as well as $\kappa=\infty$ for two reasons. 
	First, the agents may not be able to acquire perfectly accurate information. 
	Second, $F_\infty(\tau_y)$ is constant if the marginal cost is zero and the agents choose $\tau_x=\infty$ in the worst-case scenario, i.e., $F_\infty(\tau_y)=W(\infty,\tau_y)<\infty$, which can be the case for some welfare functions.

The policymaker chooses the precision of public information that maximizes the worst-case welfare, which is said to be robustly optimal or robust for short.  
\begin{definition}
The precision of public information $\tau_y^*\in\mathbb{R}_+\cup\{\infty\}$ is $\kappa$-robust if 
\[
F_\kappa(\tau_y^*)=\sup_{\tau_y}F_\kappa(\tau_y).
\]
We say that $\tau_y^*$ is robust if $\tau_y^*$ is $\kappa$-robust for all $\kappa\in \mathbb{R}_{++}\cup \{\infty\}$. 
\end{definition}

A key observation is that the worst-case cost function is linear. 
This is because if a linear cost function and a strictly convex cost function have the same marginal cost at the same equilibrium precision of private information in the second-period subgame, the total cost is greater for the former than the latter. 
Moreover, when the cost function is linear, the dispersion equals the cost by \eqref{cost and dispersion}. Therefore, the worst-case welfare has the following representation. 
\begin{lemma}\label{worst key lemma}
For each $\tau_y\geq 0$ and $\kappa\in\mathbb{R}_+\cup\{\infty\}$, 
\begin{align}
F_\kappa(\tau_y)= \inf_{\tau_x\leq \kappa} W^0(\tau_x,\tau_y), \label{infimum W}
\end{align}
where 
$W^0(\tau_x,\tau_y)\equiv \eta V(\tau_x,\tau_y)+(\zeta-1)D(\tau_x,\tau_y)$ is the expected welfare when the cost function is linear and the precision of acquired private information is $\tau_x$ in the second-period subgame.
\end{lemma}

By this lemma, if $W^0(\tau_x,\tau_y)$ is monotone in $\tau_y$ for each $\tau_x$,  
then $F_\kappa(\tau_y)$ is also monotone in $\tau_y$ for each $\kappa<\infty$. 
Thus, using the MWD in the case of exogenous private information, 
we can obtain a sufficient condition for the monotonicity of the worst-case welfare and robustness of full or no disclosure.

\begin{proposition}\label{main result robustly optimal -1}
The following holds. 

\begin{enumerate}[\em (i)]
	\item Suppose that $\eta>0$ and $\eta\geq 2(1-\alpha)(\zeta-1)/3$. Then, 
$W^0(\tau_x,\tau_y)$ is strictly increasing in $\tau_y$ for each $\tau_x$. 
Thus, $F_\kappa(\tau_y)$ is strictly increasing for each $\kappa\in\mathbb{R}_+\cup\{\infty\}$. The robust precision is $\infty$. 
	\item Suppose that $\eta<0$ and $\eta\leq 2(1-\alpha)(\zeta-1)/3$. Then, 
$W^0(\tau_x,\tau_y)$ is strictly decreasing in $\tau_y$ for each $\tau_x$. 
Thus, $F_\kappa(\tau_y)$ is strictly decreasing for each $\kappa<\infty$, while $F_\infty(\tau_y)$ is constant. The robust precision is $0$. 
\end{enumerate}
\end{proposition}

It is clear that the worst-case welfare is increasing 
if the welfare is increasing regardless of cost functions, but the converse is not true. 
For example, if $\lim_{\lambda\to\infty}\underline\eta(\zeta,\lambda)=2(1-\alpha)\zeta/3>\eta>2(1-\alpha)(\zeta-1)/3>0$, then the worst-case welfare is increasing by Proposition~\ref{main result robustly optimal -1}, so full disclosure is robust. However, no disclosure can be optimal if $\lambda$ is sufficiently large by Proposition~\ref{optimal proposition}. Such an example will be discussed in Section~\ref{Cournot games}. 

Similarly, the worst-case welfare is decreasing if the welfare is decreasing regardless of cost functions, but the converse is not true. 
For example, if $0>2(1-\alpha)(\zeta-1)/3>\eta>(1-\alpha)(\zeta-1)=\underline\eta(\zeta,0)$, then 
the worst-case welfare is decreasing by Proposition \ref{main result robustly optimal -1}, so no disclosure is robust. However, partial disclosure is optimal if $\lambda$ is sufficiently small by Proposition \ref{optimal proposition}.  

Even if $W^0(\tau_x,\tau_y)$ is not monotone in $\tau_y$, 
we can identify robust disclosure by directly calculating \eqref{infimum W}.

\begin{proposition}\label{main result robustly optimal}
The following holds. 

\begin{enumerate}[\em (i)]
	\item Suppose that $0<\eta< 2(1-\alpha)(\zeta-1)/3$. Then, $F_\kappa(\tau_y)$ is strictly increasing for each $\kappa\in\mathbb{R}_+\cup\{\infty\}$. The robust precision is $\infty$. 
	\item Suppose that $0>\eta> 2(1-\alpha)(\zeta-1)/3$ and $\kappa<\infty$. Then, $F_\kappa(\tau_y)$ is strictly increasing for $\tau_y< g(\kappa)$ and strictly decreasing for $\tau_y> g(\kappa)$, where 
\[
g(\kappa)\equiv 
-(1-\alpha)(3 \eta-2(1-\alpha)  (\zeta-1))\kappa/\eta-\tau_\theta. 
\] 
The $\kappa$-robust precision is $\max\{g(\kappa),0\}$, which equals $0$ if $g(\kappa)\leq 0$, i.e., $\kappa\leq -\tau_\theta\eta/((1-\alpha)(3 \eta-2(1-\alpha)  (\zeta-1)))$. 
If $\kappa=\infty$, the following holds.
\begin{enumerate}[\em (a)]
\item If $\eta\leq (1-\alpha)(\zeta-1)/2$, $F_\infty(\tau_y)$ is constant. 
\item If $\eta > (1-\alpha)(\zeta-1)/2$, $F_\infty(\tau_y)$ is strictly increasing. The $\infty$-robust precision is $\infty$.
\end{enumerate}
\end{enumerate}
\end{proposition}

In summary, full disclosure is robust if and only if $\eta>0$, which is true if and only if full disclosure is optimal under some linear cost function by Corollary \ref{optimal proposition corollary}.  
In this case, robust disclosure is more informative than optimal disclosure. 
If $\eta<0$ and $\eta\leq 2(1-\alpha)(\zeta-1)/3$, no disclosure is robust, where robust disclosure is not necessarily more informative than optimal disclosure. 
If $0>\eta> 2(1-\alpha)(\zeta-1)/3$, 
$\kappa$-robust precision is increasing in $\kappa$, and there is no robust precision. 
Figure \ref{figure 4.5} illustrates the classes of welfare functions identified by 
Propositions \ref{main result robustly optimal -1} and \ref{main result robustly optimal} on the $\zeta$-$\eta$ plane.

\begin{figure}
\centering
\includegraphics[width=8cm, bb=0 0 540 335]{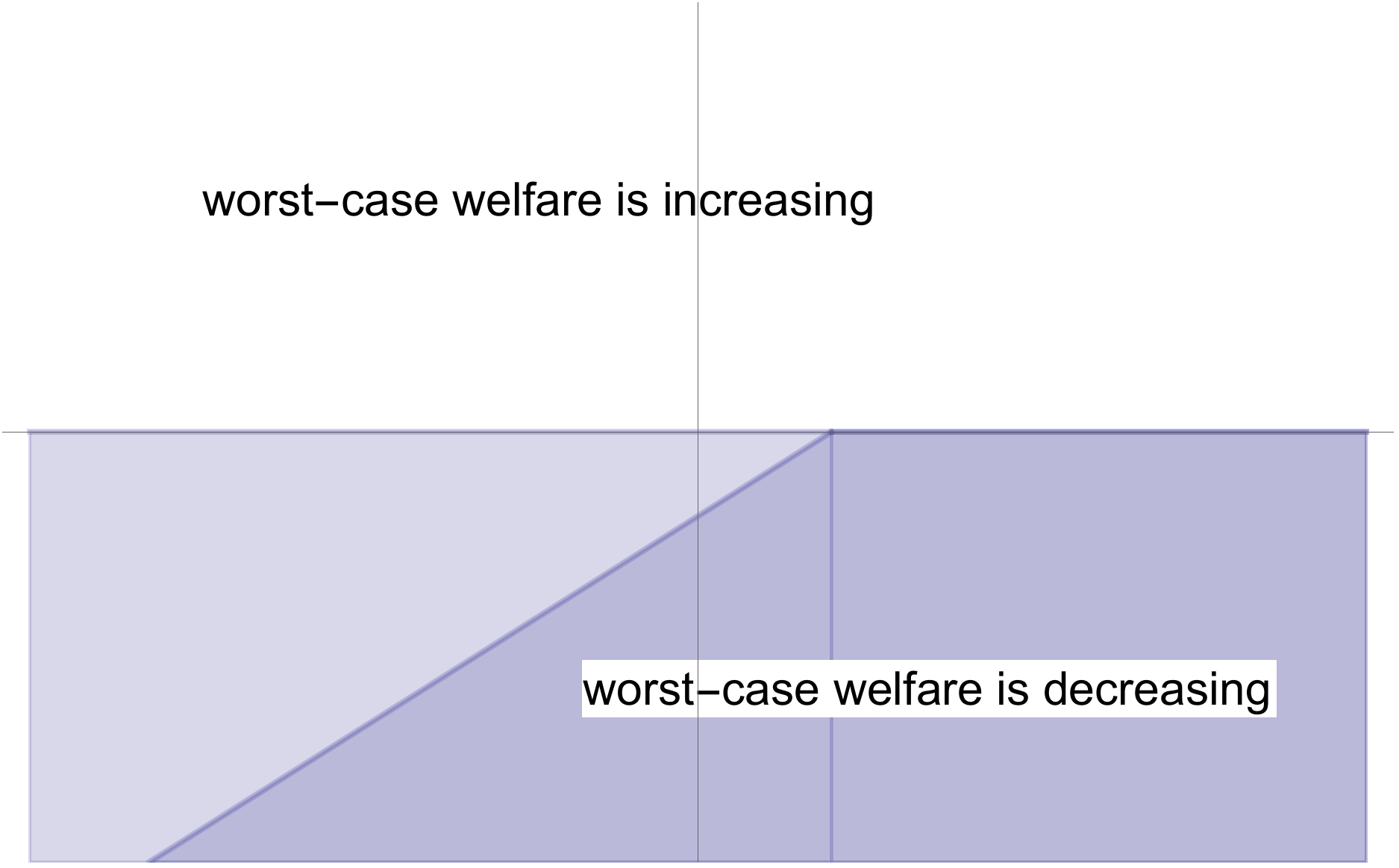}
\caption{The effects of public information on the worst-case welfare. The worst-case welfare necessarily increases with public information if $(\zeta,\eta)$ is in the upper area; it necessarily decreases if $(\zeta,\eta)$ is in the lower-right area.}
\label{figure 4.5}
\end{figure}

\section{Applications}\label{Cournot}


\subsection{Cournot games}\label{Cournot games}
Consider a Cournot game with a continuum of firms. 
Firm $i$ produces $a_i$ units of a homogeneous product at a quadratic cost $a_i^2$. 
An inverse demand function is $\theta-\delta \int a_jdj$,  
where $\delta>0$ is constant and $\theta$ is normally distributed. Then, firm $i$'s profit is 
\begin{equation}
\Big(\theta-\delta \int a_jdj\Big)a_i- a_i^2, \notag\label{Cournot payoff}
\end{equation}
which is \eqref{payoff function cont} with $\alpha=-\delta/2$ and $\beta=1/2$.
When the information cost is linear in the precision, the second-period subgame is the game studied by  \citet{lietal1987} and \citet{vives1988}. 
Let $W(\tau_x,\tau_y)$ be the net total profit, i.e., the expected total profit minus the information cost. Then, by Lemma~\ref{lemma 2},  $W(\tau_x,\tau_y)$ has a representation \eqref{welfare simple} with $(\zeta,\eta)=(1,1)$; that is, $W(\tau_x,\tau_y)=\mathrm{var}[\sigma_i]-C(\tau_x)$.

The optimal and robust precision of public information is given by the following corollary of Propositions \ref{optimal proposition} and \ref{main result robustly optimal -1}. 

\begin{corollary}\label{Cournot corollary}
Consider a Cournot game with an isoelastic information cost function \eqref{isoelastic cost function}. 
If $\lambda=0$ or $\delta<\delta^*(\lambda)\equiv 1+2\lambda^{-1}$, 
the net profit necessarily increases with public information. 
If $\delta>\delta^*(\lambda)$, 
the net profit can decrease with public information.  
Full disclosure is optimal if $\delta<\delta^{**}(\lambda)\equiv 2{\sqrt{(1+\lambda^{-1})  (1+\lambda^{-1}+\tau_\theta/\phi(0))}}  +2\lambda^{-1}$, and no disclosure is optimal if $\delta>\delta^{**}(\lambda)$. 
On the other hand, full disclosure is robust for all $\delta>0$. 
\end{corollary}

More precise public information can decrease the net profit if $\delta>\delta^*(\lambda)$, and no disclosure is optimal if $\delta>\delta^{**}(\lambda)>\delta^{*}(\lambda)$, where the price elasticity of demand is sufficiently small. 
A related result is reported in the case of exogenous private information by \citet{morrisbergemann2013} and \citet{uiyoshizawa2015}, which corresponds to the limit as $\lambda$ approaches infinity, and the intuition behind the result is essentially the same. 
When $\delta$ is large, the game exhibits strong strategic substitutability, so the ratio of the coefficient of private information to that of public information in the equilibrium strategy, $b_x/b_y=(1+\delta)\tau_x/\tau_y$, is large.  Thus, the dispersion term is dominant in the net profit, which is decreasing in the precision of public information.

Note that $\delta^{*}(\lambda)$ is decreasing in $\lambda$ with the supremum $\infty$ and the infimum $1$. 
Thus, for any fixed $\delta>0$, $\delta<\delta^{*}(\lambda)$ for sufficiently small $\lambda$, and $\delta>\delta^{*}(\lambda)$ for sufficiently large $\lambda$ if $\delta>1$.  In other words, more precise public information can be harmful if $\lambda$ is sufficiently large but beneficial if $\lambda$ is sufficiently small, which is attributed to a large reduction in the information cost under the crowding-out effect, as discussed in Section \ref{The welfare effect of public information}.

\subsection{Beauty contest games}\label{Beauty contest games}

Let $\alpha =r\in(0,1)$ and $\beta=1-r$ in (\ref{payoff function cont}).  The best response is the conditional expectation of the weighted mean of the state and the aggregate action,  
$\E[ (1-r) \theta+r\int a_jdj| x_i,y]$. 
The welfare is the negative of the mean squared error of an action from the state minus the information cost:  
\[
W(\tau_x,\tau_y)=-\E[(\sigma_i-\theta)^2]-C(\tau_x),
\]
which has a representation \eqref{welfare simple} with $(\zeta,\eta)=(1+r,1-r)$ by Lemma~\ref{lemma 2}.
This game is referred to as a beauty contest game. 
It is clear that full disclosure is optimal, but more precise public information can be harmful, as shown by \citet{morrisshin2002} in the case of exogenous private information and by  \citet{ui2014} in the case of endogenous private information.
The following result is a special case of  Proposition~\ref{optimal proposition}.

\begin{corollary}\label{beauty contest corollary}
Consider a beauty contest game  
with an isoelastic cost function \eqref{isoelastic cost function}.  
 If $r< r^*(\lambda)\equiv (\lambda/2+1)/(\lambda+1)$, the welfare necessarily increases with public information.  
If $r> r^*(\lambda)$, the welfare can decrease with public information. 
\end{corollary}

Suppose that $\lambda>0$. Then, the welfare can decrease if $r$ is sufficiently close to one because $r^*(\lambda)<1$. 
 However, if $\lambda=0$, i.e., the cost function is linear, then $r^*(\lambda)=1$, so the welfare necessarily increases for all $r\in (0,1)$, as shown by \citet{colombofemminis2008}. 
Note that $r^{*}(\lambda)$ is decreasing in $\lambda$ with the supremum $1$ and the infimum $1/2$ (see Figure~\ref{figure 6}).    
Thus, for any fixed $r\in  (1/2,1)$, $r<r^{*}(\lambda)$ for sufficiently small $\lambda$, and $r>r^{*}(\lambda)$ for sufficiently large $\lambda$; that is, more precise public information can be harmful if $\lambda$ is sufficiently large but beneficial if $\lambda$ is sufficiently small.   

\citet{colombofemminis2014} apply their result to a beauty contest game and show that the crowding-out effect can turn the social value of public information from negative to positive, but never from positive to negative. 
As discussed in Section \ref{The welfare effect of public information}, this is the case if and only if 
${MW_D}^0>0>{MW_D}^*(\tau_y)$, and indeed, 
${MW_D}^0=1-r>0$ and ${MW_D}^*(\tau_y)<0$ for sufficiently small $\tau_y$ and $r>1/2$.

We also consider the gross welfare 
$W^g(\tau_x,\tau_y)=-\E[(\sigma_i-\theta)^2]$,  which has a representation \eqref{welfare simple} with $(\zeta,\eta)=(1+r+1/(\lambda+1),1-r)$ by \eqref{gross welfare}. 
The following result on the gross welfare is a consequence of Proposition~\ref{optimal proposition}.

\begin{corollary}\label{beauty contest corollary new}
Consider a beauty contest game with an isoelastic cost function \eqref{isoelastic cost function}.  
If $r< r^g(\lambda) \equiv {\lambda }/({2 (\lambda +1)})=r^{*}(\lambda)-1/ (\lambda +1)$, the gross welfare necessarily increases with public information.  
If $r> r^g(\lambda)$, the gross welfare can decrease with public information. 
\end{corollary}


Note that $r^{g}(\lambda)$ is increasing in $\lambda$ with the infimum $0$ and the supremum $1/2$ (see Figure~\ref{figure 6}). 
Thus, for any fixed $r\in (0,1/2)$, $r>r^{g}(\lambda)$ for sufficiently small $\lambda$, and $r<r^{g}(\lambda)$ for sufficiently large $\lambda$; that is, more precise public information can be harmful to the gross welfare if $\lambda$ is sufficiently small but beneficial if $\lambda$ is sufficiently large, which is due to a small increase in the volatility caused by the crowding-out effect, as discussed in Section \ref{The welfare effect of public information}.

\begin{figure}
\centering
\includegraphics[width=8cm, bb=0 0 540 330]{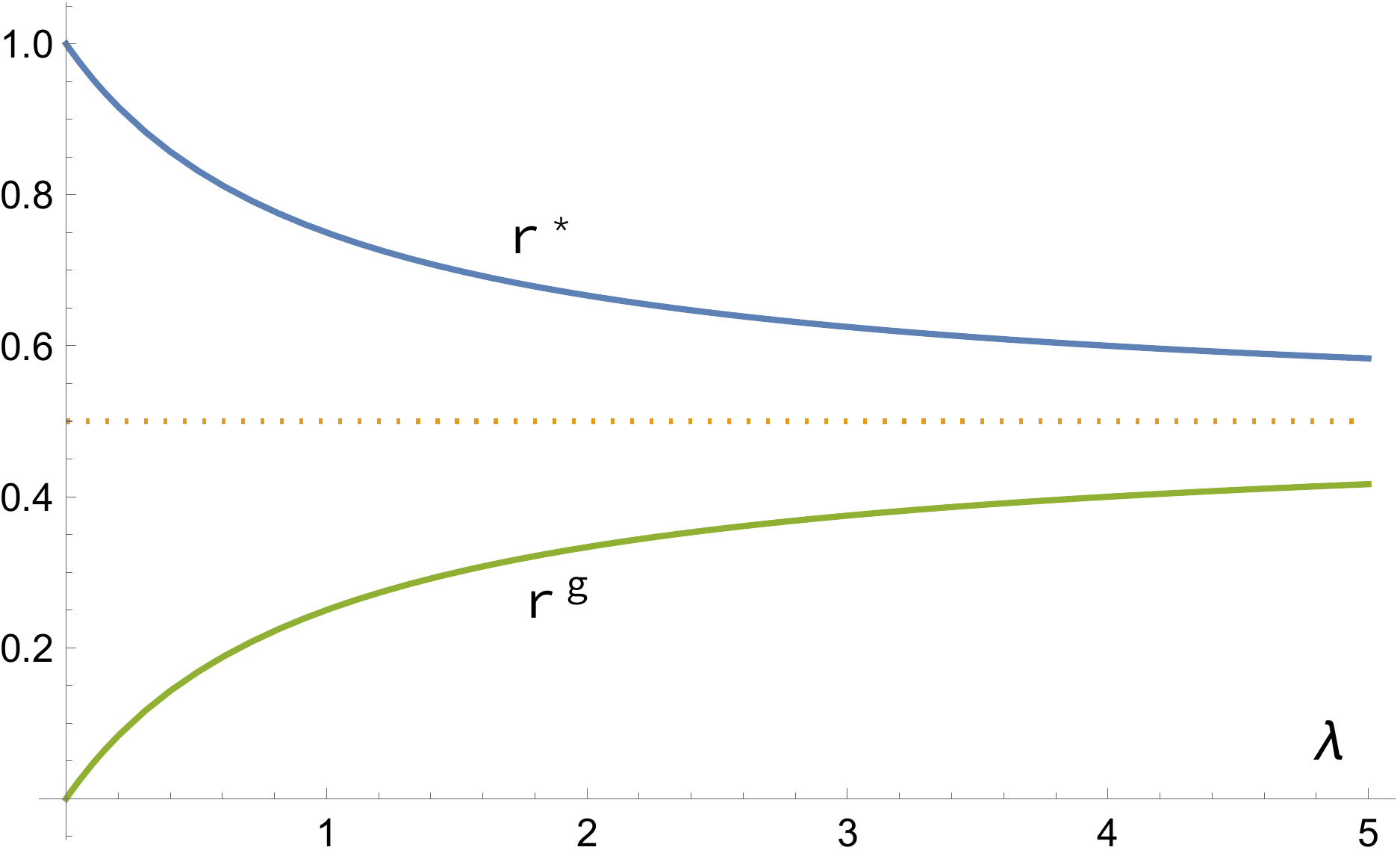}
\caption{The graphs of $r^*(\lambda)$ and $r^g(\lambda)$.}
\label{figure 6}
\end{figure}

Finally, we consider the worst-case welfare using Proposition~\ref{main result robustly optimal -1}. 
\begin{corollary}\label{beauty contest corollary robust}
The worst-case expected welfare necessarily increases with public information. 
\end{corollary}

The result gives an alternative perspective to the debate on central bank transparency prompted by \citet{morrisshin2002}.
Because their finding is presented as an anti-transparency result, it is challenged by many papers. 
\citet{svensson2006} demonstrates that more precise public information is harmful only in very special circumstances. 
Other papers extend the original model in different directions and show that more precise public information is beneficial in their extended models \citep{angeletospavan2004, hellwig2005, cornandheinemann2008, colombofemminis2008}.\footnote{\citet{jameslawler2011} show that welfare necessarily decreases with public information when a policymaker has a direct influence on payoffs as well as announces public information.}  
Corollary \ref{beauty contest corollary robust} provides a new rationale for central bank transparency: more precise public information is beneficial to the worst-case welfare. 
This is because the worst-case welfare is equal to the infimum of the expected welfare over all linear cost functions, and the expected welfare with a linear information cost necessarily increases with public information, as shown by \citet{colombofemminis2008}.

\section{Conclusion}\label{Conclusion}

This article studies the optimal and robust disclosure of public information to agents who also acquire costly private information. 
Depending on the elasticity of marginal cost, an optimal rule is qualitatively the same as the case of either the largest crowding-out effect (a linear information cost) or no crowding-out effect (exogenous private information). 
If the volatility has a positive coefficient in the welfare, either full or no disclosure is optimal, whereas only full disclosure is robust. 
Thus, a robust rule is more informative than an optimal rule, which provides an alternative rationale for central bank transparency in the context of beauty contest games. 
On the other hand, if the volatility has a negative coefficient, either no or partial disclosure is optimal and robust, where a robust rule is not necessarily more informative.

We assume the policymaker incurs no cost in providing public information throughout the analysis, but we can directly incorporate a cost of public information into our analysis. 
To study optimal costly disclosure, it is enough to compare the marginal cost and the social value of public information (i.e., the MWD times the absolute value of the marginal dispersion).  
An optimal rule under costly provision must be less informative than the optimal rule discussed in this paper. 
A detailed comparison is left for future research.

We focus on the disclosure of public information because a policymaker such as a central bank makes a public announcement in conducting policy. 
However, a policymaker may also be interested in providing private information using information technology. 
Building on our framework, we can study an optimal combination of public and private information in the following setting:  
the policymaker chooses the precision of public and private information, and each agent determines the additional precision of private information. 
A study of this model is also left for future research.

The crucial assumption in this article is the convexity of a cost function. 
Even in the case of a concave cost, the MWD is represented as the weighted average of the two special cases, and we can use it similarly to study the optimal disclosure of public information. 
Contrary to the case of a convex cost, however, 
the elasticity of marginal cost is negative, and the crowding-out effect is larger than the case of a linear cost. 
A typical concave cost is Shannon's mutual information used in the literature of rational inattention \citep{sims2003}. 
In \citet{ui2022}, we identify optimal disclosure of public information in the case of the mutual information based on \citet{rigos2020} and \citet{denti2020}.

\begin{appendices}
	
\bigskip
\appendix
\setcounter{section}{0}
\setcounter{theorem}{0}
\setcounter{lemma}{0}
\setcounter{claim}{0}
\setcounter{proposition}{0}
\setcounter{definition}{0}

\renewcommand{\theequation}{A.\arabic{equation} }
\setcounter{equation}{0}

\renewcommand{\thetheorem}{\Alph{theorem}}
\renewcommand{\thelemma}{\Alph{lemma}}
\renewcommand{\theclaim}{\Alph{claim}}
\renewcommand{\theproposition}{\Alph{proposition}}
\renewcommand{\thedefinition}{\Alph{definition}}

\section{Proofs for Section \ref{The welfare effect of public information}}

\begin{proof}[Proof of Lemma \ref{MWD special cases}]
We can obtain $MW_D(\tau_y)$ by calculating $d{V}/{d\tau_y}$, $d{D}/{d\tau_y}$, and $d{C}/{d\tau_y}$. 
Suppose that $C$ is linear. Then, 
by \eqref{linear cost}, \eqref{CV}, and \eqref{IV}, 
\begin{equation}
V(\psi_c(\tau_y),\tau_y)=\frac{\beta ^2/\tau_\theta+c  (\tau_y+\tau_\theta)-2 \beta  \sqrt{c} }{(1-\alpha)^2 },\ 
D(\psi_c(\tau_y),\tau_y)=\frac{\beta  \sqrt{c}-c (\tau_y+\tau_\theta)}{1-\alpha}.
\notag\label{linear V and D} 
\end{equation}
Thus, we have $dV(\psi_c(\tau_y),\tau_y)/d\tau_y=c/(1-\alpha)^2$ and $dD(\psi_c(\tau_y),\tau_y)/d\tau_y=dC(\psi_c(\tau_y))/d\tau_y=-c/(1-\alpha)$, by which we obtain ${MW_D}^0$.  

For arbitrary $C$, let $\tau_x=\phi(\tau_y)$. 
Then, by \eqref{CV} and \eqref{IV}, 
\begin{equation}
-\frac{\partial V(\tau_x,\tau_y)}{\partial \tau_y}/\frac{\partial D(\tau_x,\tau_y)}{\partial \tau_y}=\frac{3 (1-\alpha) \tau_x+\tau_y+\tau_\theta}{2 (1-\alpha)^2 \tau_x}, \label{MRTex1}
\end{equation}
by which we obtain ${MW_D}^*(\tau_y)$. 
\end{proof}

\begin{proof}[Proof of Proposition \ref{main proposition 1} and Corollary \ref{main corollary 1}]
By Lemma \ref{MWD special cases}, 
${MW_D}^{0}=\eta{MV_D}^{0}-\zeta+1$ with 
${MV_D}^{0}={1}/{(1-\alpha)}$, and 
${MW_D}^{*}=\eta{MV_D}^*-\zeta$ with 
\eqref{MRTex}. 
Thus, Proposition \ref{main proposition 1} and Corollary \ref{main corollary 1} are equivalent. 
We prove Corollary \ref{main corollary 1}.

We first establish \eqref{MRTgeneral}. 
Let $\tau_x=\phi(\tau_y)>0$.  
By plugging \eqref{phiinv} into \eqref{CV}, we have
\begin{align}
V(\tau_x,\phi^{-1}(\tau_x))&=
\frac{\beta^2/\tau_\theta-(1-\alpha)\tau_xC'(\tau_x) -\beta \sqrt{C'(\tau_x)} }{(1-\alpha)^2},\notag
\end{align}
which implies 
\begin{align}
\frac{d V(\phi(\tau_y),\tau_y)}{d \tau_y}
&=\frac{d V(\tau_x,\phi^{-1}(\tau_x))}{d \tau_x}\phi'(\tau_y)\notag\\
&=\left(-\frac{ C'(\tau_x)}{1-\alpha}-\tau_x C''(\tau_x)
\left(\frac{1}{1-\alpha}+
\frac{\beta }{{2(1-\alpha)^2\tau_x\sqrt{C'(\tau_x)}} }\right)\right)\phi'(\tau_y)>0.\notag\label{dV}
\end{align}
By dividing this by \eqref{dD} and using $\rho=\tau_xC''(\tau_x)/C'(\tau_x)$, we obtain 
\begin{align}
\frac{d V(\phi(\tau_y),\tau_y)}{d \tau_y}/\frac{d D(\phi(\tau_y),\tau_y)}{d \tau_y}
&=-\frac{1}{1+\rho}\left(\frac{1}{1-\alpha}+\rho
\left(\frac{1}{1-\alpha}+
\frac{\beta }{{2(1-\alpha)^2\tau_x\sqrt{C'(\tau_x)}} }\right)\right). \notag
\end{align}
This implies \eqref{MRTgeneral} 
because  
${MV_D}^0=1/(1-\alpha)$ by Lemma \ref{MWD special cases}  
and 
\begin{equation}
{MV_D}^*(\tau_y)=\frac{1}{1-\alpha}+\frac{\beta }{2 (1-\alpha)^2\tau_x\sqrt{C'(\tau_x))}},\label{MV calc 2}
\end{equation}
which we can obtain by plugging \eqref{phiinv} into \eqref{MRTex1}.

To complete the proof, it remains to show \eqref{MWD calc2}. 
By \eqref{dD}, we have 
\[
\frac{d{C(\phi(\tau_y))}}{d\tau_y}/
\frac{d D(\phi(\tau_y),\tau_y)}{d \tau_y}
=C'(\tau_x)\phi'(\tau_y)/
\frac{d D(\phi(\tau_y),\tau_y)}{d \tau_y}=\frac{1}{1+\rho},
\]
which implies \eqref{MWD calc2}.
\end{proof}


\begin{proof}[Proof of Corollary \ref{main corollary 2}]

Because ${MV_D}^*(\tau_y)> 3/(2(1-\alpha))$ by Corollary~\ref{main corollary 1}, if $\eta> 0$, then 
\begin{equation}
{MW_D}(\tau_y)> 
F(\zeta,\eta,\rho)\equiv 
\frac{{MV_D}^0+3\rho/(2(1-\alpha))}{1+\rho}\eta-\zeta+\frac{1}{1+\rho}=\frac{3\rho+2}{2(1-\alpha)(1+\rho)}\eta-\zeta+\frac{1}{1+\rho}.\notag\label{infMWD}	
\end{equation}
The unique value of $\eta$ solving $F(\zeta,\eta,\rho)=0$ is $\underline\eta(\zeta,\rho)$. 
Thus, if $\eta> \max\{\underline\eta(\zeta,\rho),0\}$, then $MW_D(\tau_y)> F(\zeta,\eta,\rho)>F(\zeta,\underline\eta(\zeta,\rho),\rho)=0$. 
Similarly, if $\eta< \min\{\underline\eta(\zeta,\rho),0\}$, then $MW_D(\tau_y)< F(\zeta,\eta,\rho)<F(\zeta,\underline\eta(\zeta,\rho),\rho)=0$. 
\end{proof}

\section{Proofs for Section \ref{The optimal disclosure of public information}}
 
\begin{proof}[Proof of Proposition \ref{optimal proposition}]
Corollary \ref{main corollary 2} implies (i) and (iii). 
We prove (ii) and (iv). 

Suppose that $\lambda=0$. 
Note that $\phi(\tau_y)>0$ if and only if $\tau_y<\beta/\sqrt{c}-\tau_\theta$ by 
\eqref{linear cost}. 
Thus, when $\tau_y<\beta/\sqrt{c}-\tau_\theta$, $dW/d\tau_y> 0$ if and only if ${MW_D}^0> 0$ by Lemma \ref{MWD special cases}, where ${MW_D}^0=(\eta-\underline\eta(\zeta,0))/(1-\alpha)$ because $\underline\eta(\zeta,0)=
(1-\alpha)(\zeta-1)$.
When $\tau_y\geq \beta/\sqrt{c}-\tau_\theta$, $dW/d\tau_y> 0$ if and only if $\eta> 0$ because $W(\phi(\tau_y),\tau_y)=W(0,\tau_y)=\eta V(0,\tau_y)$. 
This establishes the proposition in the case of $\lambda=0$.

Suppose that $\lambda>0$. 
Plugging \eqref{MV calc 2} into \eqref{MRTgeneral} and \eqref{MWD calc2}, 
we have
\begin{align}
MW_D(\tau_y)
&=\eta\left(\frac{1}{1-\alpha}+\frac{\lambda\beta }{2 (1+\lambda)(1-\alpha)^2\phi(\tau_y)\sqrt{C'(\phi(\tau_y))}}\right)-\zeta+\frac{1}{1+\lambda}\label{MWD calc 2}
\end{align}
because $\rho=\lambda$. If $MW_D(\tau_y)=0$ has a solution $\bar\tau_y$,  $\bar\tau_x=\phi(\bar\tau_y)$ satisfies 
\begin{align}
\bar\tau_x\sqrt{C'(\bar\tau_x)}
&=
c^{1/2}(\bar\tau_x)^{(\lambda+2)/2}=\frac{\beta   \lambda \eta}{2 \left(1-\alpha\right) \left(\left(1-\alpha\right) ((1+\lambda)\zeta -1)-(1+\lambda)\eta \right)}
\label{taux-star}	
\end{align}
by \eqref{MWD calc 2}, which implies \eqref{taux-star'}. 
In addition, 
\begin{align}
\bar\tau_z
=\phi^{-1}(\bar\tau_x)+\tau_\theta
=\bar\tau_x\left(-(1-\alpha)+{\beta}/{\left(\bar\tau_x\sqrt{C'(\bar\tau_x)}\right)}\right)
=\bar\tau_x\frac{(1-\alpha ) (3 \lambda +2)(\underline\eta(\zeta,\eta)-\eta  )}{ \lambda\eta  }\notag 
\end{align}
by \eqref{phiinv} and \eqref{taux-star}, 
which is \eqref{phiinv'}. 
Note that $\bar\tau_z>0$ if and only if either $0<\eta<\underline\eta(\zeta,\lambda)$ or $0>\eta>\underline\eta(\zeta,\lambda)$. 
In each case, $MW_D(\tau_y)=0$ has a solution $\bar\tau_y=\bar\tau_z-\tau_\theta$ if $\bar\tau_z>\tau_\theta$, which establishes (ii) and (iv). 
\end{proof}

\section{Proofs for Section \ref{regardless of costs}}

\begin{proof}[Proof of Proposition \ref{corollary main result 2}]
We prove the first half of the proposition by showing (i) $\Rightarrow$ (iii) $\Rightarrow$ (ii) because (ii) $\Rightarrow$ (i) is obvious. 
The proof for the second half is similar, so we omit it.

We show (i) $\Rightarrow$ (iii). 
Suppose (i) holds: for $\tau_y>0$,  
${MW_D}(\tau_y)>0$ for any cost function with $\phi(\tau_y)>0$. 
By Proposition \ref{main proposition 1}, ${MW_D}(\tau_y)$ equals ${MW_D}^0$ if $\rho=0$, and ${MW_D}(\tau_y)$ can be arbitrarily close to ${MW_D}^*(\tau_y)$ when $\rho$ is sufficiently large. 
Thus, we must have ${MW_D}^0>0$ and ${MW_D}^*(\tau_y)\geq 0$ for any cost function. 
Note that ${MW_D}^*(\tau_y)\geq 0$ for any cost function if and only if $\eta\geq \max\{2(1-\alpha)\zeta/3,0\}=\max\{\underline\eta(\zeta,\infty),0\}$ because the infimum of ${MV_D}^*(\tau_y)$ is $3/(2(1-\alpha))$ and the supremum is infinity by Corollary \ref{main corollary 1}.
Note also that ${MW_D}^0>0$ if and only if $\eta>(1-\alpha)(\zeta-1)=\underline\eta(\zeta,0)$ by Lemma~\ref{MWD special cases}. These conditions are summarized as (iii).

We show (iii) $\Rightarrow$ (ii). 
Suppose (iii) holds. 
Then, ${MW_D}^0>0$ and ${MW_D}^*(\tau_y)\geq 0$ by the above argument. Thus, ${MW_D}(\tau_y)>0$ if $\phi(\tau_y)>0$ by Proposition \ref{main proposition 1}, which implies (ii).  
\end{proof}

\section{Proofs for Section \ref{robust disclosure}}

We first prove Lemma \ref{worst key lemma}.
\begin{proof}[Proof of Lemma \ref{worst key lemma}]
Fix $\tau_y$. Let $C\in \Gamma_\kappa$ be such that $\phi_{C}(\tau_y)=\tau_x\in (0,\kappa]$. 
Note that $C'(\tau_x)= c\equiv \beta^2/((1-\alpha)\tau_x+\tau_y+\tau_\theta)^2$ by \eqref{FOC1}. 
Define $\tilde C\in \Gamma$ by $\tilde C(\tau_x')=c\tau_x'$ for $\tau_x'\leq \tau_x$ and 
$\tilde C(\tau_x')=C(\tau_x')-C(\tau_x)+c\tau_x$ for $\tau_x'\geq \tau_x$. 
Note that $\tilde C'(\tau_x')=C'(\tau_x')$ for all $\tau_x'\geq \tau_x$. 
This implies that $\tilde C\in \Gamma_\kappa$ and $\phi_{\tilde C}(\tau_y)=\phi_{C}(\tau_y)=\tau_x$. 
Note that $C'(\tau_x')\leq C'(\tau_x)=c$ for $\tau_x'\leq\tau_x$ since $C$ is convex. 
Thus, it holds that $C(\tau_x)=\int_0^{\tau_x}C'(t) dt\leq \int_0^{\tau_x}c dt=c\tau_x=\tilde C(\tau_x)$, and 
\begin{align*}
W(\phi_{C}(\tau_y),\tau_y)
\geq\eta V(\tau_x,\tau_y)+\zeta D(\tau_x,\tau_y)-\tilde C(\tau_x)
= \eta V(\tau_x,\tau_y)+\zeta D(\tau_x,\tau_y)-c\tau_x=W^0(\tau_x,\tau_y)
\end{align*}
by \eqref{cost and dispersion}. 
This holds with equality if $\phi_{C}(\tau_y)=0$; that is, $W(0,\tau_y)=W^0(0,\tau_y)=\eta V(0,\tau_y)$ since $D(0,\tau_y)=0$.
Consequently, we have $\inf_{C\in \Gamma_\kappa}W(\phi_{C}(\tau_y),\tau_y)=\inf_{\tau_x\leq\kappa}W^0(\tau_x,\tau_y)$ because, for each $\tau_x\leq \kappa$, there exists $C\in \Gamma_\kappa$ with $\phi_C(\tau_y)=\tau_x$.
\end{proof}

To prove Propositions \ref{main result robustly optimal -1} and \ref{main result robustly optimal}, we also use the following lemma, 
which is calculated from \eqref{CV} and \eqref{IV} (a similar result appears in \citet{uiyoshizawa2015}). 

\begin{lemma}\label{key partial derivative}
The following holds.
\begin{align}
{\partial W^0}/{\partial  \tau_y }\gtrless 0 & \ \Leftrightarrow \ 
{   (3 \eta-2 (1-\alpha ) (\zeta-1)    )\tau_x+\eta   \left(\tau_y+\tau_{\theta }\right)/(1-\alpha ) }\gtrless 0,\notag\\
{\partial W^0}/{\partial  \tau_y }= 0 & \ \Leftrightarrow \ \tau_y=g(\tau_x)\equiv -(1-\alpha)(3 \eta-2(1-\alpha)  (\zeta-1))\tau_x/\eta-\tau_\theta,\notag\\
{\partial W^0}/{\partial \tau_x}\gtrless 0 & \ \Leftrightarrow \  
{  (2 \eta-(1-\alpha)  (\zeta-1) )\tau_x+(\zeta -1) (\tau_y+\tau_\theta)}\gtrless 0,
\notag\\
{\partial W^0}/{\partial \tau_x}= 0 & \ \Leftrightarrow \  
\tau_x=f(\tau_y)\equiv-(\zeta-1)(\tau_y+\tau_\theta)/(2 \eta-(1-\alpha)  (\zeta-1)).
\notag
\end{align}
\end{lemma}



\begin{proof}[Proof of Proposition \ref{main result robustly optimal -1}]

Suppose that $\eta>0$ and $\eta\geq 2(1-\alpha)(\zeta-1)/3$. 
Consider $F_\kappa$ with $\kappa<\infty$. 
By Lemma~\ref{key partial derivative}, ${\partial W^0}/{\partial \tau_y }>0$.  
For $\tau_y^0<\tau_y^1$, there exist $\tau_x^0$ and $\tau_x^1$ such that 
$W^0(\tau_x^0,\tau_y^0)=\inf_{\tau_x\leq \kappa} W^0(\tau_x,\tau_y^0)=F_\kappa(\tau_y^0)$
and $W^0(\tau_x^1,\tau_y^1)=\inf_{\tau_x\leq \kappa} W^0(\tau_x,\tau_y^1)=F_\kappa(\tau_y^1)$ by Lemma \ref{worst key lemma}
since $W^0$ is continuous and $\kappa<\infty$. 
Then,  $F_\kappa(\tau_y^0)=W^0(\tau_x^0,\tau_y^0)\leq W^0(\tau_x^1,\tau_y^0)<W^0(\tau_x^1,\tau_y^1)=F_\kappa(\tau_y^1)$. 
Thus, $F_\kappa$ is strictly increasing. 
Next, we consider $F_\infty$. 
Note that $\partial W^0/\partial\tau_x>0$ if and only if $\tau_x>f(\tau_y)$ by Lemma~\ref{key partial derivative} since 
\[
2\eta-(1-\alpha)(\zeta-1)\geq 
\begin{cases}
2(\eta-2(1-\alpha)(\zeta-1)/3)>0 &\text{ if }\zeta-1> 0,\\
2\eta>0 & \text{ if }\zeta-1\leq 0.
\end{cases}
\]
Thus, $F_\infty(\tau_y)=\inf_{\tau_x\leq\infty}W^0(\tau_x,\tau_y)<W^0(\infty,\tau_y)$. 
Consequently, $F_\infty(\tau_y^0)<F_\infty(\tau_y^1)$ for $\tau_y^0<\tau_y^1$ 
because there exists $\kappa<\infty$ such that 
$F_\infty(\tau_y^0)=F_\kappa(\tau_y^0)$ and $F_\infty(\tau_y^1)=F_\kappa(\tau_y^1)$; that is, 
$F_\infty$ is strictly increasing.

Suppose that $\eta<0$ and $\eta\leq 2(1-\alpha)(\zeta-1)/3$. 
By a similar argument, $F_\kappa$ with $\kappa<\infty$ is strictly decreasing. 
Consider $F_\infty$. 
Note that $\partial W^0/\partial\tau_x<0$ if and only if $\tau_x>f(\tau_y)$ by Lemma~\ref{key partial derivative} since $2\eta-(1-\alpha)(\zeta-1)<0$. 
Thus, $F_\infty(\tau_y)=\min\{W^0(0,\tau_y),W^0(\infty,\tau_y)\}=W^0(\infty,\tau_y)$ because we can directly verify that $W^0(0,\tau_y)>W^0(\infty,\tau_y)$. 
That is, $F_\infty$ is constant.
\end{proof}


\begin{proof}[Proof of Proposition \ref{main result robustly optimal}]

Suppose that $0<\eta< 2(1-\alpha)(\zeta-1)/3$. 
Note that $\zeta-1>0$. 
Then, 
$F_\kappa(\tau_y)=\min\{W^0(0,\tau_y),W^0(\kappa,\tau_y)\}$ by Lemma \ref{key partial derivative}. 
Because
	\begin{align}W^0(\kappa,\tau_y)-W^0(0,\tau_y)=\frac{\beta^2 \kappa (\eta  \kappa+(\zeta-1)(\tau_y+\tau_\theta))}{(\tau_y+\tau_\theta) ((1-\alpha)\kappa+\tau_y+\tau_\theta)^2}>0,\notag\label{check corner 1}
	\end{align}
we have $F_\kappa(\tau_y)=W^0(0,\tau_y)=\eta V(0,\tau_y)$, which is strictly increasing in $\tau_y$.

Suppose that $0>\eta> 2(1-\alpha)(\zeta-1)/3$. 
Note that $\zeta-1<0$. 
If $2 \eta\leq (1-\alpha)  (\zeta-1) $, then ${\partial W^0}/{\partial \tau_x}<0$ by Lemma \ref{key partial derivative}. Thus, $F_\kappa(\tau_y)=W^0(\kappa,\tau_y)$, and $F_\infty(\tau_y)=W^0(\infty,\tau_y)$ is constant. Assume that $\kappa<\infty$. 
Then, $\partial F_\kappa(\tau_y)/\partial \tau_y=\partial W^0(\kappa,\tau_y)/\partial\tau_y\gtrless 0$ for $\tau_y\lessgtr g(\kappa)$ by Lemma \ref{key partial derivative}. Thus, $\tau_y=\max\{g(\kappa),0\}$ is $\kappa$-robust. 

If $2 \eta>(1-\alpha)  (\zeta-1)$, then ${\partial W^0}/{\partial \tau_x}\lessgtr 0$ for $\tau_x \lessgtr f(\tau_y)$ and $f(\tau_y)>0$ by Lemma \ref{key partial derivative}.  Thus, $F_\kappa(\tau_y)=W^0(f(\tau_y),\tau_y)$ if $f(\tau_y)\leq\kappa$ and $F_\kappa(\tau_y)=W^0(\kappa,\tau_y)$ if $f(\tau_y)>\kappa$. By direct calculation, we have 
\[
\frac{d W^0(f(\tau_y),\tau_y)}{d\tau_y}=
\frac{\beta ^2 (2 \eta-(1-\alpha)  (\zeta -1))^2}{4 (1-\alpha)^2 (\tau_y+\tau_\theta)^2 (\eta-(1-\alpha)  (\zeta -1))}>0
\]
because $\eta>(1-\alpha)  (\zeta -1)/2>(1-\alpha)  (\zeta -1)$. 
That is, $F_\kappa(\tau_y)$ is strictly increasing if $f(\tau_y)\leq\kappa$. 
On the other hand, $f^{-1}(\kappa)-g(\kappa)=-{2 \kappa (\eta-(1-\alpha) (\zeta -1))^2}/{((\zeta -1) \eta )}<0$ because $(\zeta -1) \eta >0$. This implies that $\partial W^0(\kappa,\tau_y)/\partial\tau_y>0$ if $f^{-1}(\kappa)\leq \tau_y\leq g(\kappa)$ 
because $\partial W^0/\partial\tau_y\gtrless 0$ for $\tau_y \lessgtr g(\tau_x)$ by Lemma \ref{key partial derivative}. 
Therefore, $\tau_y=\max\{g(\kappa),0\}$ is $\kappa$-robust. 
\end{proof}

\section{Proofs for Section \ref{Cournot}}

\begin{proof}[Proof of Corollary \ref{Cournot corollary}]
When $(\zeta,\eta)=(1,1)$, 
$\underline{\eta}(\zeta,\lambda)={(\delta +2) \lambda }/({3 \lambda +2})$ and $\eta-\underline{\eta}(\zeta,\lambda)=({\lambda(1-\delta) +2})/({3 \lambda +2})$. 
Thus,  $\eta> \max\{\underline{\eta}(\zeta,\lambda), 0\}$ if and only if $\lambda=0$ or $\delta<\delta^*(\lambda)$, 
and $0<\eta< \underline{\eta}(\zeta,\lambda)$ if and only if $\delta>\delta^*(\lambda)$. 
In addition, $W(\phi(\infty),\infty)\gtrless W(\phi(0),0)$ if and only if $\delta\lessgtr \delta^{**}(\lambda)$ since 
\[
W(\phi(\infty),\infty)-W(\phi(0),0)=
\frac{\phi(0) \left(-\rho \delta ^2 +4 \delta +4\rho +8\right)+4(\rho +1) \tau_\theta}{(\delta +2)^2 (\rho +1) ((\delta +2)\phi(0)+2\tau_\theta)^2}.
\]
Finally, full disclosure is robust because $(\zeta,\eta)=(1,1)$ satisfies (i) in Proposition~\ref{main result robustly optimal -1}. 
\end{proof}

\begin{proof}[Proof of Corollary \ref{beauty contest corollary}]
When $(\zeta,\eta)=(1+r,1-r)$,  $\underline{\eta}(\zeta,\eta)={2 (1-r) ((1+r)\lambda +r)}/({3 \lambda +2})>0$, so $\eta> \underline{\eta}(\zeta,\eta)$ if and only if $r< r^*(\lambda)$. 
\end{proof}

\begin{proof}[Proof of Corollary \ref{beauty contest corollary new}]
When $(\zeta,\eta)=(1+r+1/(\lambda+1),1-r)$,  $\underline{\eta}(\zeta,\eta)=2 (1+\lambda) (1-r^2)/{(3 \lambda +2)}$, so $\eta> \underline{\eta}(\zeta,\eta)$ if and only if $r< r^g(\lambda)$. 
\end{proof}

\end{appendices}


\begin{thebibliography}{99}

\bibitem[Angeletos and Pavan(2004)]{angeletospavan2004}  Angeletos, G.-M., Pavan, A., 2004. Transparency of information and coordination in economies with investment complementarities. Amer. Econ. Rev. 94, 91--98.

\bibitem[Angeletos and Pavan(2007)]{angeletospavan2007} 
Angeletos, G.-M., Pavan, A., 2007. Efficient use of information and social value of information. Econometrica 75, 1103--1142.



\bibitem[Bergemann and Morris(2013)]{morrisbergemann2013} Bergemann, D., Morris, S., 2013.  Robust predictions in games with incomplete information. Econometrica 81, 1251--1308. 

\bibitem[Bergemann and Morris(2019)]{morrisbergemann2017} Bergemann, D., Morris, S., 2019.  Information design: A unified perspective. J.\ Econ.\ Lit.\ 57, 44--95.  

\bibitem[Bizzotto et al.(2020)]{nizzottoetal2020}
Bizzotto, J., R\"udiger, J., Vigier, A., 2020. 
Testing, disclosure and approval. 
J. Econ. Theory 187, 105002.



\bibitem[Bloedel and Segal(2020)]{bloedelsegal2020}
Bloedel, A., Segal, I., 2020. 
Persuading a rationally inattentive agent. Working
Paper.




\bibitem[Cai and Dong(2021)]{caidong2021}
Cai, Z., Dong, F., 2021. 
Public disclosure and private information acquisition: A global game approach. 
Working paper.


\bibitem[Colombo and Femminis(2008)]{colombofemminis2008} 
Colombo, L., Femminis, G.,  2008. 
The social value of public information with costly information acquisition.
Econ. Letters 100, 196--199.


\bibitem[Colombo et al.(2014)]{colombofemminis2014} Colombo, L., Femminis, G.,  Pavan, A., 2014. Information acquisition and welfare. Rev.\ Econ. Stud. 81, 1438--1483.

\bibitem[Cornand and Heinemann(2008)]{cornandheinemann2008} Cornand, C., Heinemann, F., 2008. Optimal degree of public information dissemination. Econ. J. 118, 718--742.




\bibitem[Denti(2020)]{denti2020}
Denti, T., 2020. Unrestricted Information Acquisition. Working paper. 

\bibitem[Enconomist(2004)]{economist2004}
Economist, 2004. It's not always good to talk. The Economist July 24, 76.


\bibitem[Dworczak and Pavan(2020)]{dworczakpavan2020}
Dworczak, P., Pavan, A., 2020.
Preparing for the worst but hoping for the best: Robust (Bayesian) persuasion. Working paper.




\bibitem[H\'ebert and La'O(2021)]{hebertlao2021}
H\'ebert, B., La'O, J., 2021. 
Information acquisition, efficiency, and non-fundamental volatility. 
NBER Working Paper 26771.


\bibitem[Hellwig(2005)]{hellwig2005}  Hellwig, C., 2005. Heterogeneous information and the welfare effects of public information disclosures. Working paper. 

\bibitem[Hellwig and Veldkamp(2009)]{hellwigveldkamp2009}
Hellwig, C., Veldkamp, L., 2009. 
Knowing what others know: coordination motives in information acquisition. 
Rev.\ Econ.\ Stud.\ 76, 223--251.

\bibitem[Hu and Weng(2021)]{huweng2021}
Hu, J., Weng, X., 2021. 
Robust persuasion of a privately informed receiver. Econ.\ Theory 72, 909--953.

\bibitem[James and Lawler(2011)]{jameslawler2011} James, J.\ G., Lawler, P., 2011. Optimal policy intervention and the social value of public information. Amer.\ Econ.\ Rev. 101, 1561--1574.


\bibitem[Kamenica and Gentzkow(2011)]{kamenicagentzkow2011}
Kamenica, E., Gentzkow, M., 2011. Bayesian persuasion. 
 Amer.\ Econ.\ Rev. 101, 2590--2615.
 
\bibitem[Kamenica(2019)]{kamenica2019} Kamenica, E., 2019. Bayesian persuasion and information design. Annu.\ Rev.\ Econ.\  11, 249--72



\bibitem[Kosterina(2021)]{kosterina2021}
Kosterina, S., 2021. Persuasion with unknown beliefs. Forthcoming in Theor. Econ.


\bibitem[Leister(2020)]{leister2020}
Leister, M., 2020. 
Information acquisition and welfare in network games. 
Games Econ.\ Behav. 122, 453--475.


\bibitem[Li et al.(1987)]{lietal1987}
Li, L., McKelvey, R. D., Page, T.,  1987. 
Optimal research for Cournot oligopolists. 
J.\ Econ.\ Theory 42, 140--166.


\bibitem[Lipnowski et al.(2020)]{lipnowskietal2020a}
Lipnowski, E., Mathevet, L., Wei, D., 2020. Attention management. 
 Amer.\ Econ.\ Rev.: Insights 2, 17--32.



\bibitem[Mackowiak and Wiederholt(2009)]{mackowiakwiederholt2009}
Mackowiak, B.,Wiederholt, M., 2009. Optimal sticky prices under rational inattention. Amer.\ Econ.\ Rev.\ 99, 769--803.



\bibitem[Matyskov\'a and Montes(2021)]{matyskovamontes2021}
Matyskov\'a, L., Montes, A., 2021.
Bayesian persuasion with costly information acquisition. 
Working paper.

\bibitem[Morris and Shin(2002)]{morrisshin2002} Morris, S., Shin, H. S., 2002.  Social value of public information. Amer. Econ. Rev. 92, 1521--1534. 


\bibitem[Myatt and Wallace(2012)]{myattwallace2012}
Myatt, D. P., Wallace, C., 2012. 
Endogenous information acquisition in coordination games. 
Rev.\ Econ. Stud. 79, 340--374. 


\bibitem[Myatt and Wallace(2015)]{myattwallace2015}
Myatt, D. P., Wallace, C., 2015. 
Cournot competition and the social value of information. 
J. Econ. Theory 158, 466--506.

\bibitem[Myatt and Wallace(2019)]{myattwallace2019}
Myatt, D. P., Wallace, C., 2019. 
Information acquisition and use by networked players. 
J.\ Econ.\ Theory 182, 360--401.


\bibitem[Radner(1962)]{radner1962} 
Radner, R., 1962. Team decision problems. Ann. Math. Stat. 33, 857--881.

\bibitem[Rigos(2020)]{rigos2020} 
Rigos, A., 2020. Flexible Information acquisition in large coordination games. Working paper.


\bibitem[Sims(2003)]{sims2003} 
Sims, C.\ A., 2003. Implications of rational inattention. J.\ Monet.\ Econ.\ 50, 665--690
 
 
\bibitem[Svensson(2006)]{svensson2006} Svensson, L. E. O., 2006. Social value of public information: comment: Morris and Shin (2002) is actually pro-transparency, not con.   Amer. Econ. Rev. 96, 448--452.





\bibitem[Ui(2014)]{ui2014}Ui, T., 2014. 
The social value of public information with convex costs of information acquisition. 
Econ.\ Letters 125, 249--252.

\bibitem[Ui(2020)]{ui2020}
Ui, T., 2020. 
LQG information design. Working paper. 

\bibitem[Ui(2022)]{ui2022}
Ui, T., 2022. Impacts of public information on flexible information acquisition. Working paper. 

\bibitem[Ui and Yoshizawa(2013)]{uiyoshizawa2012}
Ui, T., Yoshizawa, Y.,  2013. 
Radner's theorem on teams and games with a continuum of players. 
Econ.\ Bulletin 33, 72--77. 

\bibitem[Ui and Yoshizawa(2015)]{uiyoshizawa2015}
Ui, T., Yoshizawa, Y.,  2015. 
Characterizing social value of information. J.\ Econ.\ Theory 158, 507--535.



\bibitem[Vives(1988)]{vives1988}
Vives, X., 1988. Aggregation of information in large Cournot markets. Econometrica 56,  851--876.

\bibitem[Vives(2008)]{vives2008} Vives, X., 2008. Information and Learning in Markets: The Impact of Market Microstructure.  Princeton Univ.\ Press. 

\bibitem[Wei(2021)]{wei2021}
Wei, D., 2021. Persuasion under costly learning. J. Math. Econ. 94, 102451.

\bibitem[Wong(2008)]{wong2008}
Wong, J., 2008. Information acquisition, dissemination, and transparency of monetary policy. Can.\ J.\ Econ.\ 41, 46--79.





\end{thebibliography}
\end{document}